\documentclass[journal]{IEEEtran}
\ifCLASSINFOpdf

\else
 \fi
 \usepackage{amsmath}
\usepackage{amsfonts}
\usepackage{amsthm,mathrsfs}
\usepackage{dsfont}
\usepackage{amssymb}
\usepackage{amstext}
\usepackage{color, epsf, epsfig, graphicx, psfrag, subfigure}
\usepackage{color,cite,times,latexsym,ifthen}

\newcommand{\SNR}{\mathsf{SNR}}
\newcommand{\INR}{\mathsf{INR}}
\newcommand{\GDoF}{\mathsf{GDoF}}
\newcommand{\DoF}{\mathsf{DoF}}
\newcommand{\rank }{\mathrm{rank}}
\newcommand{\FB }{\mathrm{FB}}

\newcommand{\bS}{\mathbf{S}}
\newcommand{\cC}{\mathcal{C}}
\newcommand{\E}{\mathbb{E}}
\newcommand{\F}{\mathbb{F}}
\newcommand{\sR}{\mathscr{R}}
\newcommand{\Rs}{R_{\mathrm{sym}}}
\newcommand{\Cs}{\tilde{C}_{\mathrm{sym}}}
\newcommand{\cV}{\mathcal{V}}
\newcommand{\bc}{\mathbf{c}}
\newcommand{\bd}{\mathbf{d}}
\newcommand{\bx}{\mathbf{x}}
\newcommand{\by}{\mathbf{y}}
\newcommand{\bz}{\mathbf{z}}
\newcommand{\bs}{\mathbf{s}}

\newcommand{\tz}{\tilde{z}}
\newcommand{\g}{\tilde{g}}

\renewcommand{\t}[1]{\mathbf{Tx}_{#1}}
\renewcommand{\r}[1]{\mathbf{Rx}_{#1}}
\newcommand{\cgc}[3]{\frac{#1}{#2}\log\left(1+ #3 \right)}
\newcommand{\cg}[1]{\frac{1}{2}\log\left(1+ #1 \right)}

\newtheorem{lm}{Lemma}
\newtheorem{remark}{Remark}
\newtheorem{thm}{Theorem}
\newtheorem{example}{Example}

\begin{document}

\title{On the Feedback Capacity of the Fully Connected $K$-User Interference Channel}

\author{Soheil~Mohajer,  
Ravi Tandon, ~\IEEEmembership{Member, IEEE}, 
and 
H. Vincent Poor, \IEEEmembership{Fellow,  IEEE}
  \thanks{Manuscript received October 21 2011; revised May 24 2012; and accepted December 14 2012. }
  \thanks{{\footnotesize 
      This paper was presented in part at IEEE International Symposium on Information Theory (ISIT), Cambridge, USA, 2012. 
            
      Soheil Mohajer  was with the Department of Electrical Engineering, 
      Princeton University, Princeton, NJ, USA. He is now with  the
      Department of Electrical Engineering and Computer Sciences,
      University of California at Berkeley, Berkeley, CA, USA 
      (E-mail: mohajer@eecs.berkeley.edu).
      
      Ravi Tandon  was with the Department of Electrical Engineering, 
      Princeton University, Princeton, NJ, USA. He is now with the
     Department of Electrical and Computer Engineering, Virginia Tech, 
     Blacksburg, VA, USA. (E-mail: tandonr@vt.edu).
     
	 H. Vincent Poor is with the Department of Electrical Engineering, Princeton
     University, Princeton, NJ, USA. (E-mail: poor@princeton.edu).
     
     The work was supported in part by the Air Force Office of Scientific 
     Research under MURI Grant FA-$9550$-$09$-$1$-$0643$ and in 
     part by the DTRA under Grant HDTRA-$07$-$1$-$0037$. The work 
     of Soheil Mohajer is partially supported by The Swiss National Science 
     Foundation under Grant PBELP2-$133369$.}}  
     \thanks{Copyright (c) 2012 IEEE. Personal use of this material is permitted.  However, permission to use this material for any other purposes must be obtained from the IEEE by sending a request to pubs-permissions@ieee.org.}}

\markboth{IEEE Transactions on Information Theory,~Vol.~x, No.x~, Month~20xx}%
{Shell \MakeLowercase{\textit{et al.}}: Bare Demo of IEEEtran.cls for Journals}

\maketitle

\begin{abstract}
	The symmetric $K$ user interference channel with fully connected topology is considered, in which {\sf (a)} each receiver suffers interference from all other $(K-1)$ transmitters, and {\sf (b)} each transmitter has causal and noiseless feedback from its respective receiver.  The number of generalized degrees of freedom ($\GDoF$) is characterized in terms of $\alpha$, where the interference-to-noise ratio ($\INR$) is given by $\INR=\SNR^\alpha$. It is shown that the per-user $\GDoF$ of this network is the same as that of the $2$-user interference channel with feedback, except for $\alpha=1$, for which existence of feedback does not help in terms of $\GDoF$. The coding scheme proposed for this network, termed cooperative interference alignment, is based on two key ingredients, namely, interference alignment and interference decoding. Moreover, an approximate characterization is provided for the symmetric feedback capacity of the network, when the $\SNR$ and $\INR$ are far apart from each other. 
\end{abstract}

\section{Introduction}
Wireless networks with multiple pairs of transceivers are quite common in modern communications, notable examples being  wireless local area networks (WLANs) and cellular networks. Multiple independent flows of information share a common medium in such multiple unicast wireless networks. The broadcast and superposition nature of the wireless medium introduces complex signal interactions between multiple competing flows. In contrast to the point-to-point wireless channel, where  a noisy version of a single transmitted signal is  received at a given receiver, a combination of various wireless signals are observed at receivers in multiple unicast systems. In such scenarios, each decoder has to deal with all interfering signals in order to decode  its intended message. Managing such interfering signals in a multi-user network is a long standing and fundamental problem in wireless communication.

The simplest example in this category is the $2$-user  interference channel \cite{HK:81}, in which  two transmitters with independent messages wish to communicate with their respective receivers over the wireless transmission medium.  Even for this simple $2$-user  network, the complete information-theoretic characterization of the capacity region has been  open for several decades. To study more general networks, there is a clear need for a deep understanding  and perhaps develop novel interference management techniques. 

Although the exact characterization of the capacity region of the $2$-user  Gaussian interference channel is still unknown,  several inner and outer bounds are known. These bounds are very useful in the sense of providing  an approximate characterization when there exists a guarantee on the gap between them. This approach has resulted in an approximate characterization, within one bit, by Etkin, Tse, and Wang in  \cite{ETW:08} as well as Telatar and Tse in \cite{TT:IF:07}. This characterization includes upper bounds for the capacity of the network, as well as encoding/decoding strategies based on Han-Kobayashi scheme \cite{HK:81}, which perform close to optimal. Moreover, it has been shown that the gap between the fundamental information-theoretic bounds and what can be achieved using the proposed schemes is provably small. Therefore, the capacity  can be approximated within a narrow range, although the exact region is still unknown. 

A similar approximate characterization (with a larger gap) for this problem is developed in \cite{BT08}, in which both coding scheme 
and bounding techniques are devised by studying the problem under the \emph{deterministic} model. This framework, introduced by  Avestimehr, Diggavi, and Tse  in \cite{ADT:11:J}, focuses on complex signal interactions in a wireless network by ignoring the randomness of the noise. Recently, it has  been successfully applied to several problems, providing valuable insights for the more practically relevant Gaussian problems.

Several interference management techniques have been proposed for operating over more complex interference networks.  Completely or partially decoding and removing interference (interference suppression) when it is strong and treating it as noise when it is weak are perhaps the most widely used schemes. More sophisticated schemes such as interference  alignment \cite{MMK08, CJ08} have been proposed recently. However, it still remains to be seen whether the capacity of general interference networks can be  achieved with any combination of these techniques. 

It is well known that feedback does not increase the capacity of point-to-point discrete memoryless channels \cite{zero:56}. However, 
feedback is beneficial in improving the capacity regions of more complex networks (see \cite{EGK-book} and references therein). The effects of feedback on the capacity region of the interference channel have been studied in several  papers. Feedback coding schemes for $K$-user Gaussian interference networks have been developed by Kramer in \cite{kramer2002feedback}. Outer bounds for the $2$-user interference channel with generalized feedback have been derived in \cite{GastparKramer:06} and \cite{tandon:11}. The effect of feedback on the capacity of the $2$-user interference channel is studied in \cite{CJ-FIC:08}, where it is shown that feedback provides multiplicative gain in the capacity at high signal-to-noise ratio ($\SNR$), when the interference links are much stronger than the direct links. The entire feedback capacity region of the  $2$-user Gaussian interference channel has been characterized within a $2$ bit gap by Suh and Tse in \cite{SuhTse11}. This includes all regimes of interference, and finite and asymptotic regimes of $\SNR$. The gap between the capacity of the channel with and without feedback can be arbitrarily large for certain channel parameters. The key technique for the strong interference regime is to use the feedback links to create an artificial path  from each transmitter to its respective receiver through the other nodes in the network. For instance, the message intended for $\r{1}$, can be sent either through the direct link $\t{1} \rightarrow \r{1}$, or the cyclic path $\t{1}\rightarrow \r{2} \rightarrow \t{2} \rightarrow \r{1}$. In particular, the advantage of such artificial paths can be clearly understood when the cross links are much stronger than the direct links (e.g., the strong interference regime).  This observation becomes very natural by studying the problem under the deterministic framework. 

The first extension of \cite{SuhTse11} to a multi-user setting is the $K$-user cyclic interference channel with feedback, where each receiver's signal is interfered with only one of its neighboring transmitters,  in a cyclic fashion. The effect of feedback  on the capacity region of this network is addressed in \cite{TMP:11}. It is shown that although feedback improves the symmetric capacity  of the $K$-user interference channel,  the improvement in symmetric capacity per user vanishes as $K$ grows. The intuitive reason behind this result is that the configuration of the network allows only one cyclic path, which has to be shared between all pair of transceivers. The amount of information that can be conveyed through this path does not scale with $K$, and therefore the gain for each user scales inverse linearly with $K$. 

In another extreme, each transmit signal may be corrupted by all the other signals transmitted by the other base stations. This model is appropriate for a network with densely located nodes, where everyone hears everyone else.  This network, which we call \emph{the fully connected $K$-user interference channel} (FC-IC), is another generalization of the $2$-user interference channel. Fig.~\ref{fig:general-model} shows the fully connected IC with feedback for $K=3$ users. In this paper, we study the FC-IC network with feedback, and for simplicity, we consider a symmetric network topology,  where all the direct links (from each transmitter to its respective receiver) have the same gain, and similarly, the gain of all cross (interfering) links are identical. A similar setting without feedback has been studied by Jafar and Vishwanath in \cite{JV:K-IF:10}, where the number of symmetric degrees of freedom is characterized. An approximate sum capacity of this network is recently found by Ordentlich \emph{et al.}  in \cite{KIC-SumCap:Erez}. In this paper, the impact of feedback is studied for the $K$-user FC-IC. The main contribution of this paper is to show that feedback can arbitrarily improve the performance of the network, and in contrast to the cyclic case \cite{TMP:11}, it \emph{does scale} with the number of users in the systems. In particular, except for the intermediate interference regime where the signal-to-noise ratio is equal to the interference-to-noise ratio ($\SNR=\INR$), the effect of interference from $K-1$ users is as if there were only one interfering transmitter in the network. This is analogous to the result of \cite{CJ08}, where it is shown that the number of per-user degrees of freedom of the $K$-user fading interference channel, is the same as if there were only $2$ users in the network. 

\begin{figure}[t]
\centering
\subfigure[A cellular interference network.]{
\includegraphics[height=6cm]{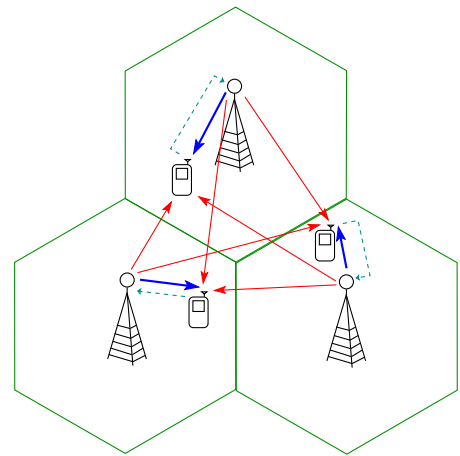}
\label{fig:subfig1}
}\vspace{1cm}
\subfigure[Interference network with feedback.]{
 	\psfrag{t1}[Bc][Bc]{$\t{1}$}
	\psfrag{t2}[Bc][Bc]{$\t{2}$}
	\psfrag{t3}[Bc][Bc]{$\t{3}$}
	\psfrag{r1}[Bc][Bc]{$\r{1}$}
	\psfrag{r2}[Bc][Bc]{$\r{2}$}
	\psfrag{r3}[Bc][Bc]{$\r{3}$}
\includegraphics[height=6cm]{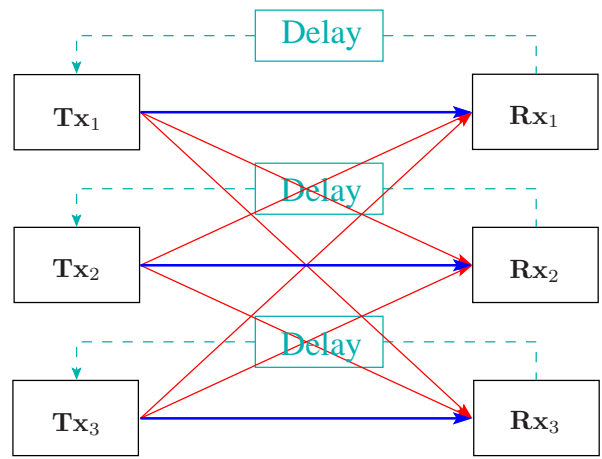}
\label{fig:subfig2}
}
\label{fig:general-model}
\caption{A cellular network with three base stations and three clients in (a),  simplified and modelled as the network in (b).}
\end{figure}

In order to get the  maximal benefit of feedback, we propose a novel encoding scheme, called cooperative interference alignment, which combines two well-known interference management techniques, namely, interference alignment and interference decoding. More precisely, the encoding  at the transmitters is such that all the interfering signals are aligned at each receiver. However, a fundamental difference between our approach and the standard interference alignment approach is that we need to decode interference to be able to remove it from the received signal, while the aligned interference is usually suppressed in standard approaches. A challenge here, which makes this problem fundamentally different from the $2$-user inference channel, is that the interference is a combination of $(K-1)$ interfering messages, and decoding all  of them induces strict bounds on the rate of the interfering messages. However, each transmitter does not need to decode all the interfering messages individually, instead, upon receiving feedback, it only decodes the combination of them that corrupts the intended signal of interest. To this end, we propose using a common structured code, which has the property that the summation of codewords of different users is still another codeword from the same codebook. Lattice codes \cite{ZSE-lattice} are a suitable choice to satisfy this desired property. This idea is similar to that used in \cite{nazer2007computation} and \cite{jafarian-lattice09}.

The rest of this paper is organized as follows. First, we formally present the model, introduce notation, and state the problem in Section~\ref{sec:problem}. The main result of the paper is presented in Section~\ref{sec:result}. Before proving the result for the Gaussian network, we study the problem under the deterministic model in Section~\ref{sec:det}, where we characterize the exact feedback capacity of the deterministic network. Based on the insight and intuition obtained by analysis of the deterministic network, we present the converse proof and the coding scheme for the Gaussian network in Sections~\ref{sec:G-Ach} and \ref{sec:G-UB}, respectively. Having the approximate feedback capacity of the network, we derive the generalized degrees of freedom with feedback in Section~\ref{sec:dof}. We further extend the result of the paper, and study the case of global feedback, where each transmitter receives feedback from all receivers in Section~\ref{sec:GF}, and finally, conclude the paper in Section~\ref{sec:con}. In order to make the paper easily readable,  some of the technical proofs are postponed to the appendices.  Parts of this work
have been presented in \cite{MTP:ISIT12}.

\section{Problem Statement}
\label{sec:problem}
In this work we consider a network with $K$ pairs of transmitter/receivers. Each transmitter $\t{k}$ has a message $W_k$ that it wishes to send  to its respective receiver $\r{k}$. The signal transmitted by each transmitter is corrupted by  the interfering signals sent by other transmitters, and received at the receiver. This can be mathematically modelled as
\begin{align}
y_k(t) = \sqrt{\SNR} x_k(t) + \sum_{\begin{subarray}{c}   i=1 \\ i\neq k \end{subarray}}^K \sqrt{\INR} x_i(t) + z_k(t),
\label{eq:channel}
\end{align}
where $x_k$ and $y_k$ are  the signals transmitted and received by $\t{k}$ and $\r{k}$, respectively, and $z_k\sim\mathcal{N}(0,1)$ is an additive white Gaussian noise. All transmitting powers are constrained to $1$, i.e., $\E[x_k^2]\leq 1$, for $k=1,\dots,K$. We assume a symmetric network, where all the cross links have the same gain ($\INR$), and the gains of the all the direct link ($\SNR$) are identical.

There is a perfect feedback link from each receiver to its respective transmitter. Hence, at each time instance, each transmitter generates each transmitting signal based on its own message as well as the output sequence observed at its receiver over the past time instances, i.e., 
\begin{align}
x_{kt}=g_{kt} (W_k, y_{k1}, y_{k2},\dots, y_{k(t-1)})= g_{kt}(W_k, y_k^{t-1}),
\label{eq:encoding}
\end{align}
where we use shorthand notation $y_k^{t-1}=(y_{k1}, y_{k2},\dots, y_{k(t-1)})$ to indicate the output sequence observed at $\r{k}$ up to time $t-1$.

A rate tuple $(R_1,R_2,\dots,R_K)$ is called achievable if there exists a family of codebooks with block length $T$ with proper power and corresponding encoding/decoding functions such that the average decoding error probability  tends to zero for all users as $T$ increases. We denote the set of all achievable rate tuples by $\sR$.  In the high signal to noise ratio regime, the performance of wireless networks is measured in terms of the number of degrees of freedom, that is the pre-log factor in the expression of the capacity in terms of $\SNR$. We consider the generalized degrees of freedom ($\GDoF$) for this network in the presence of feedback. Since the problem is parametrized in terms of two growing factors\footnote{The notion of degrees of freedom ($\DoF$)  captures the asymptotic behavior of the capacity, where the transmit power grows to infinity. However, this forces all channels to be equally strong, i.e., all the power of all received signals from different links grow at the same rate. Therefore, it is not very insightful towards finding optimal transmission schemes  when some signals are significantly stronger or weaker than others. The generalized degrees of freedom which allows different rate of growth for $\SNR$ and $\INR$ is more useful metric in such scenarios. We refer the reader to \cite{jafar2011interference} for a comprehensive discussion on these metrics.}, namely $\SNR$ and $\INR$,  we use the standard parameter $\alpha$ (as in \cite{ETW:08} and \cite{JV:K-IF:10}) to capture the growth rate of $\INR$ in terms of $\SNR$. More formally, we define
\begin{align}
\alpha=\frac{\log \INR}{\log \SNR},
\label{eq:def:alpha}
\end{align}
and the \emph{per-user} generalized degrees of freedom as
\begin{align}
d(\alpha)= \frac{1}{K} \limsup_{\SNR\rightarrow\infty} \frac{\max_{(R,\dots,R)\in \sR}\sum_{k=1}^K R_k(\SNR,\alpha)}{\frac{1}{2} \log \SNR}.
\end{align}
It is worth mentioning that the half factor appears in the denominator since we are dealing with real signals. Our primary goal is to characterize the generalized degrees of freedom of the $K$-user interference channel with output feedback.

As mentioned earlier, the $\GDoF$ characterizes the performance of the network in the asymptotic $\SNR$ regime. However, in order to study practical networks, capacity is a more accurate measure to capture the performance. In order to consider such a high resolution analysis, we define the symmetric capacity of the network, that is 
\[
\Rs = \max_{(R,\dots,R)\in \sR} R.
\]
In this work we are interested in characterizing $\Rs$ for the $K$-user interference channel with feedback. Although finding the exact symmetric capacity is extremely difficult, we make progress on this problem, and approximately characterize the capacity when the $\SNR$ and $\INR$ are not close to each other, that is when $\alpha$ (defined in \eqref{eq:def:alpha}) is not equal to $1$. To this end, we derive outer bounds and propose coding schemes for the network, and show that the gap between the achievable rate and the outer bound is a function \emph{only} of $K$, the number of users in the network, and is independent of $\SNR$ and $\INR$.

\section{Main Results}
\begin{figure*}[t]
\begin{center}
 	\psfrag{1}[Bc][Bc]{$2$}
 	\psfrag{2}[Bc][Bc]{$1$}
	\psfrag{3}[Bc][Bc]{$\frac{2}{3}$}
	\psfrag{4}[Bc][Bc]{$\frac{1}{2}$}
	\psfrag{d}[Bc][Bc]{$d(\alpha)$}
	\psfrag{a}[Bc][Bc]{$\alpha$}
	\psfrag{5}[Bc][Bc]{$\frac{1}{K}$}
	\psfrag{A}[Bc][Bc]{\begin{footnotesize}$K$-user/ w. FB\end{footnotesize}}
	\psfrag{B}[Bc][Bc]{\begin{footnotesize}$K$-user/ no FB\end{footnotesize}}
\includegraphics[width=0.7\textwidth]{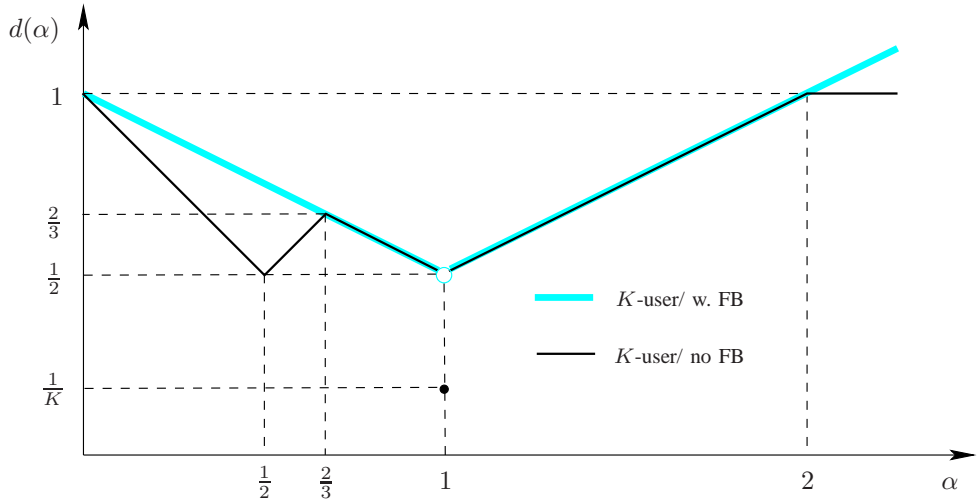}
\end{center}
\caption{The per-user generalized degrees of freedom for the $K$-user interference channel.}
\label{fig:dof}
\end{figure*}

\label{sec:result}
In this section we present the main results of this paper. The first theorem characterizes the generalized degrees of freedom of the $K$-user FC-IC with feedback. 
\begin{thm}
For the $K$-user fully connected interference channel (FC-IC) with output feedback, the per-user $\GDoF$ is given by
\begin{align*}
d_{\FB}(\alpha) = \left\{ \begin{array}{ll}
1-\frac{\alpha}{2} & \alpha<1 \ \textrm{(weak interference)}\\
\textrm{not well-defined} & \alpha=1 \\
\frac{\alpha}{2} & \alpha>1 \ \textrm{(strong interference)}
\end{array}
\right.
\end{align*}
\label{thm:DoF}
\end{thm}
We present the proof of Theorem~\ref{thm:DoF} in Section~\ref{sec:dof}. Note that the theorem above does not characterize the $\GDoF$ for $\alpha=1$. In fact for this regime, the $\GDoF$ is not well-defined and can get different values, depending on mutual growth of $\SNR$ and $\INR$. We refer the interested reader to Section~\ref{sec:dof} for a detailed discussion.

In order to demonstrate the benefit gained by output feedback, we present the following theorem from \cite{JV:K-IF:10}, which characterizes the $\GDoF$ for the FC-IC without feedback.

\begin{thm}[\cite{JV:K-IF:10}, Theorem~3.1]
The per-user $\GDoF$ for the $K$-user interference channel without feedback is given by
\begin{align*}
d_{\textrm{No } \FB}(\alpha) = \left\{ \begin{array}{ll}
1-\alpha & 0\leq \alpha\leq \frac{1}{2} \ \textrm{(noisy interference)}\\
\alpha & \frac{1}{2}\leq \alpha \leq \frac{2}{3}\ \textrm{(weak interference)}\\
1-\frac{\alpha}{2} & \frac{2}{3}\alpha<1 \ \textrm{(moderate interference)}\\
\frac{1}{K} & \alpha=1 \\
\frac{\alpha}{2} & 1< \alpha\leq 2 \ \textrm{(strong interference)}\\
1 & \alpha>2 \ \textrm{(very strong interference).}
\end{array}
\right.
\end{align*}
\label{thm:DoF-no-FB}
\end{thm}

The generalized degrees of freedom of the $K$-user interference channel with/without feedback are illustrated in Figure~\ref{fig:dof}. As derived in \cite{JV:K-IF:10}, the $\GDoF$ for the $K$-user no feedback case, is similar to that of $2$-user case \cite{ETW:08}, except for $\alpha=1$. Similarly, here we show that for the channel with feedback, the $\GDoF$ for the $K$-user case is the same as that of the $2$-user channel \cite{SuhTse11}, except for $\alpha=1$. At this particular point, the $\GDoF$ can be bounded from below and above by $\frac{1}{K}$ and $\frac{1}{2}$, respectively.

The following theorem characterizes the approximate capacity of the channel for arbitrary signal-to-noise ratio.
\begin{thm}
The symmetric capacity of the $K$ user interference channel with feedback with\footnote{A similar result can be shown when $\INR/\SNR \notin (1-\delta_1,1+\delta_2)$ where $\delta_1,\delta_2>0$ are constants. In that case the gap between the achievable rate and the upper bound may depend on $\delta_1$ and $\delta_2$. We refer the interested reader to the discussion at the end of Section~\ref{sec:dof} on the capacity behavior at $\INR \approx \SNR$. } 
\[
\frac{\INR}{\SNR} \notin \left(\frac{1}{2}, 2\right) 
\]
can be approximated by 
\begin{align}
\Cs \hspace{-1pt}\triangleq\hspace{-1pt} \frac{1}{4} \log (1+ \SNR + \INR )\hspace{-1pt}+\hspace{-1pt} \frac{1}{4} \log\hspace{-1pt} \left(\hspace{-1pt}1+\frac{\SNR}{1+\INR}\hspace{-1pt}\right)\hspace{-2pt}. 
\label{eq:C-G}
\end{align}
More precisely, the symmetric capacity is upper bounded by $C_{\mathrm{sym}}\leq \Cs+\frac{K-1}{4}+\frac{1}{2}\log K$. 
Moreover, there exists a coding scheme that can support any rate satisfying $\Rs\leq \Cs-\frac{1}{4}\log 16K^2(K+1)$.
\label{thm:Cap-Gaussian}
\end{thm}
We will present the achievability part of Theorem~\ref{thm:Cap-Gaussian} in Section~\ref{sec:G-Ach}. The proof of the converse part can be found in Section~\ref{sec:G-UB}.

\section{The Deterministic Model}
\label{sec:det}
In this section we study the problem of interest in a deterministic framework introduced in \cite{ADT:11:J}. The key point in this model  is to focus on signal interactions instead of the additive noise, and obtain insight about both coding schemes and outer bounds for the original problem. 

The intuition behind this approach is that the noise is modelled by a deterministic operation on the received signal which splits the received signal into a completely useless part and a completely noiseless part. The part of the received signal below the noise level is completely useless since it is corrupted by noise. However, the part above the noise level is assumed to be not affected by noise and can be used to retrieve information.

Let $p$ be any prime number and $\F$ be the finite field over the set $\{0,1,\dots,p-1\}$ with sum and product operations modulo $p$.  Moreover, define 
\begin{align*}
n=\lfloor \log_p \SNR\rfloor \quad \textrm{and}\quad  m=\lfloor \log_p \INR\rfloor.
\end{align*}
Each received signal can be mapped into a $p$-ary stream. Let $X_k\in\F^q$ and $Y_k\in \F^q$ be the $p$-ary expansion of the transmit and received signal by user $k$, respectively, where $q=\max\{m,n\}$. The shift linear deterministic channel model for this network can be written as 
\begin{align}
Y_k = D^{q-n} X_k + \sum_{i\neq k} D^{q-m} X_i,
\label{eq:ch-model-DET}
\end{align}
where all the operations are performed modulo $p$. Here, $D$ is the shift matrix, defined as
\begin{align*}
D=\begin{bmatrix}
0 & 0 &  0 & \cdots & 0 & 0\\
1 & 0 &  0 & \cdots & 0 & 0\\
0 & 1 & 0 & \cdots & 0 & 0\\
\vdots & & & \ddots & & \vdots\\
0 & 0 & 0 & \cdots & 1 & 0
\end{bmatrix}_{q\times q}.
\end{align*}

The following theorem characterizes the symmetric capacity of the deterministic network introduced above. In the rest of this section, we prove this theorem by first deriving an upper bound on the symmetric capacity, and then proposing coding schemes for different interference regimes. The ideas arising in this section will be later used when we focus on the Gaussian network in Sections~\ref{sec:G-UB} and \ref{sec:G-Ach}.

\begin{thm}
The symmetric feedback capacity of the linear deterministic  $K$-user fully connected interference channel with parameters $n$ and $m$ is given by
\begin{align}
\Rs=\left\{ \begin{array}{ll}
n-\frac{m}{2} & n>m \ \textrm{(weak interference)},\\
\frac{n}{K} &  m=n,\\
\frac{m}{2} & n<m\ \textrm{(strong interference)}.
\end{array}
\right.
\end{align}
\label{thm:det}
\end{thm}

\begin{remark}
From the rate expression in Theorem~\ref{thm:det} one can easily see that the normalized feedback capacity of the channel under the linear deterministic model is given by 
\begin{align*}
\frac{\Rs}{n}=\left\{ \begin{array}{ll}
1-\frac{1}{2}\left(\frac{m}{n}\right) & \frac{m}{n}<1, \\
\frac{1}{K} &  \frac{m}{n}=1,\\
\frac{1}{2} \left(\frac{m}{n}\right) & \frac{m}{n}>1,
\end{array}
\right.
\end{align*}
which is analogous to the $\GDoF$ expression in Theorem~\ref{thm:DoF}, by noting that $m/n$ is analogous to $\alpha$ for the Gaussian setting. 
\end{remark}

We may also study a generalized version of the symmetric model introduced in \eqref{eq:ch-model-DET}. As we will see in the rest of this section, the symmetric topology the current model allows a simple interference alignment at the receivers. A natural generalization of this model assigns a random sign to each channels, while channel gains are kept symmetric. More precisely, in this model, called  \emph{quasi-symmetric $K$-user fully connected interference channel}\footnote{We wish to thank the anonymous reviewer for suggesting this model.},  the gain of all the direct links are identical and all the cross links have identical gains, but each link has a random sign which captures random phase in the Gaussian model. Note that, without loss of generality, we may assume the sign of all direct links are positive, and write the channel model as 
\begin{align}
Y_k = D^{q-n} X_k + \sum_{i} \lambda_{ki} D^{q-m} X_i,
\label{eq:model-det-qsym}
\end{align}
where $\lambda_{ki}\in\{-1,+1\}$ for $i\neq k$ captures the sign of the cross link from $\t{i}$ to $\r{k}$. 
The following theorem states that a similar result as Theorem~\ref{thm:det} holds for $3$-user network.

\begin{thm}
The symmetric feedback capacity of the linear deterministic  quasi-symmetric $3$-user fully connected interference channel introduced in \eqref{eq:model-det-qsym} with parameters $n$ and $m$ is given by
\begin{align*}
\Rs=\left\{ \begin{array}{ll}
n-\frac{m}{2} & n>m \ \textrm{(weak interference)},\\
\frac{n}{3} &  m=n \textrm{ and  $\Lambda+\mathbf{I}$ is singular},\\
\frac{n}{2} &  m=n \textrm{ and  $\Lambda+\mathbf{I}$ is non-singular},\\
\frac{m}{2} & n<m\ \textrm{(strong interference)},
\end{array}
\right.
\end{align*}
where $\Lambda$ is the channel sign matrix with $\Lambda_{ij}=\lambda_{ij}$ for $i\neq j$ and $\Lambda_{ii}=0$ for $i=1,2,3$.
\label{thm:det-qsym}
\end{thm}
In the following we present an example to illustrate the reason for loss in $\GDoF$ for singular $\Lambda+\mathbf{I}$. The proof of Theorem~\ref{thm:det-qsym} can be found in Appendix~\ref{app:det-qsym}. We will also show that this result can be generalized to arbitrary $K$ provided that $\Lambda$ satisfies certain conditions. Extension of this result to the quasi-symmetric Gaussian channel would be straight-forward from the coding scheme in Section~\ref{sec:G-Ach}, and we skip it in sake of brevity. 

 \begin{figure*}[t]
\begin{center}
 	\psfrag{t1}[Bc][Bc]{$\t{1}$}
	\psfrag{t2}[Bc][Bc]{$\t{2}$}
	\psfrag{t3}[Bc][Bc]{$\t{3}$}
	\psfrag{r1}[Bc][Bc]{$\r{1}$}
	\psfrag{r2}[Bc][Bc]{$\r{2}$}
	\psfrag{r3}[Bc][Bc]{$\r{3}$}
	\psfrag{a1}[Bc][Bc]{$a_1$}
	\psfrag{a2}[Bc][Bc]{$a_2$}
	\psfrag{a3}[Bc][Bc]{$a_3$}
	\psfrag{a4}[Bc][Bc]{$a_5$}
	\psfrag{b1}[Bc][Bc]{$b_1$}
	\psfrag{b2}[Bc][Bc]{$b_2$}
	\psfrag{b3}[Bc][Bc]{$b_3$}
	\psfrag{b4}[Bc][Bc]{$b_5$}
	\psfrag{c1}[Bc][Bc]{$c_{1}$}
	\psfrag{c2}[Bc][Bc]{$c_{2}$}
	\psfrag{c3}[Bc][Bc]{$c_{3}$}
	\psfrag{c4}[Bc][Bc]{$c_5$}
	\psfrag{T1}[Bc][Bc]{$T=1$}
	\psfrag{T2}[Bc][Bc]{$T=2$}
	\psfrag{aa2}[Bc][Bc]{$a_{2}$}
	\psfrag{aa3}[Bc][Bc]{$a_{3}+(b_1+c_1)$}
	\psfrag{aa4}[Bc][Bc]{$a_4$}
	\psfrag{aa5}[Bc][Bc]{$2a_1+(b_1+c_1)+a_5$}
	\psfrag{bb2}[Bc][Bc]{$b_{2}$}
	\psfrag{bb3}[Bc][Bc]{$b_{3}+(a_1+c_1)$}
	\psfrag{bb4}[Bc][Bc]{$b_4$}
	\psfrag{bb5}[Bc][Bc]{$2b_1+(a_1+c_1)+b_5$}
	\psfrag{cc2}[Bc][Bc]{$c_{2}+(a_1+b_1)$}
	\psfrag{cc3}[Bc][Bc]{$c_{3}+(a_2+b_2)$}
	\psfrag{cc4}[Bc][Bc]{$c_4$}
	\psfrag{cc5}[Bc][Bc]{$2c_1+(a_1+b_1)+c_5$}
	\psfrag{aaa2}[Bc][Bc]{$(b_1+c_1)$}
	\psfrag{aaa3}[Bc][Bc]{$a_4$}
	\psfrag{bbb2}[Bc][Bc]{$(a_1+c_1)$}
	\psfrag{bbb3}[Bc][Bc]{$b_4$}
	\psfrag{ccc2}[Bc][Bc]{$(a_1+b_1)$}
	\psfrag{ccc3}[Bc][Bc]{$c_4$}
	\psfrag{aa4}[Bc][Bc]{$a_4$}
	\psfrag{aa5}[Bc][Bc]{$2a_1+(b_1+c_1)+a_5$}
\includegraphics[width=0.85\textwidth]{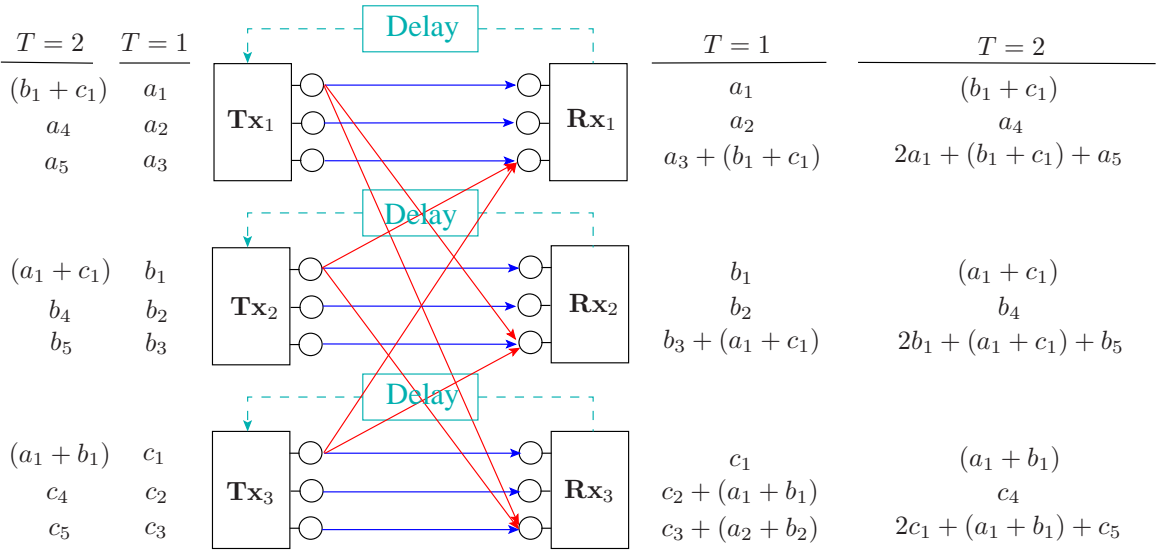}
\end{center}
\caption{Coding scheme for the linear deterministic model in the weak interference regime, for $K=3$, $n=3$, and $m=1$.}
\label{fig:det_weak}
\end{figure*}

\begin{example}
Consider a network with $K=3$ users, and sign matrix given by 
\begin{align*}
\Lambda=\begin{bmatrix}
0 & -1 & 1\\
1 & 0 &-1 \\
1 & -1 & 0 
\end{bmatrix}. 
\end{align*}
It is clear that the first and third rows of $\Lambda+\mathbf{I}$ are identical, and hence this matrix is singular. The channel model for this network for $m=n$ can be written as
\begin{align*}
Y_1&=X_1-X_2+X_3,\\
Y_2&=X_1+X_2-X_3,\\
Y_3&=X_1-X_2+X_3,
\end{align*}
in which $Y_1=Y_3$. Consider an arbitrary reliable coding scheme with block length $T$ for this network. Having the output of $\r{1}$ over the whole block, one can find 
\[
Y_1^T\rightarrow W_1 \rightarrow X_1^T=g_1^T(W_1, Y_1^T).
\] 
Similarly 
\[Y_1^T\rightarrow Y_3^T=Y_1^T \rightarrow W_3 \rightarrow X_3^T=g_3^T(W_3, Y_3^T).\] 
Therefore,  
\[X_2^T=-Y_1^T+X_1^T+X_3^T,\] 
and 
\[Y_2^T=X_1^T+X_2^T-X_3^T\] 
can be found from $Y_1^T$, and finally $W_2$ can be decoded from $Y_2^T$. In other words, having $Y_1^T$, all three messages can be decoded, i.e., 
\begin{align*}
R_1+R_2+R_3 &\leq H(W_1,W_2,W_3) \\
&\leq I(W_1+W_2+W_3; Y_1^T)+T\epsilon_T \\
&\leq H(Y_1^T) +T\epsilon_T \leq T(n+\epsilon_T),
\end{align*}
which results in $\Rs\leq n/3$, which is the same rate claimed in Theorem~\ref{thm:det-qsym}.  Achievability of this rate using time-sharing scheme is clear. 
\end{example}

\subsection{Encoding Scheme}
In the following we present a transmission scheme that can achieve the rate claimed in Theorem~\ref{thm:det}.  We first demonstrate the proposed scheme in two examples with specific parameters, through which the basic ideas and intuitions are transparent. Although generalization of the proposed coding strategy for arbitrary $n$ and $m$ is straight-forward, we present the scheme and its analysis in Appendix~\ref{app:det-ach} in sake of completeness. \\[-1mm]

\paragraph{Weak Interference Regime $(m<n)$}

The goal is to achieve $\Rs=n-\frac{m}{2}$ bits per user. We propose an encoding that operates on  a block of length $2$. The basic idea can be seen from Fig.~\ref{fig:det_weak}, wherein the coding scheme is demonstrated for $n=3$ and $m=1$. 

For these specific parameters, we have $\Rs=5/2$. As it is shown in Fig.~\ref{fig:det_weak}, the proposed coding scheme is able to convey five intended symbols from each transmitter to its respective receiver in two channel uses. The information symbols intended for $\r{1}$ are denoted by $a_1,a_2,a_3,a_4, a_5$. Each transmitter sends three fresh symbols in its first channel use. Receivers get two interference-free symbols, and one more equation, including their intended symbol as well as interference. The output signals are sent to the transmitters over the feedback link, in order to be used for the next transmission. In the second channel use, each transmitter forwards the interfering parts of its received feedback on its top level. The two lower levels will be used to transmit the remaining fresh symbols. 

Now, consider the received signals at $\r{1}$ in  two channel uses. It has received $6$ linearly independent equations, involving $7$ variables, which seems to be unsolvable at the first glance. However, we do not need to decode $b_1$ and $c_1$ individually. Instead, we can solve the system of linear of equations in $a_1$, $a_2$, $a_3$, $a_4$, $a_5$, and $(b_1+c_1)$,  which can be solved for the intended variables.  Hence, a per-user rate of $5/2$ symbols/channel-use is achievable with feedback. \\[-1mm]

\paragraph{Strong Interference Regime $(m > n)$} In this section we present an encoding scheme which can support a symmetric rate of $\Rs=\frac{m}{2}$. Again we focus on specific parameters, $n=1$ and $m=3$, which implies $\Rs=3/2$. 

As shown in Fig.~\ref{fig:det-strong}, the proposed coding strategy delivers three intended symbols to each receiver in two channel uses. In the first channel use, each transmitter sends its fresh symbols to its respective receiver. However, due to the strong interference, receivers are not able to decode any part of their intended symbols, and can only send their received signals to their respective transmitters through the feedback links. Each transmitter then removes its own contribution from the received signal, and forwards the remaining over the second channel use. Similar to the weak interference regime, at the end of the transmission each receiver  has $6$ equations, involving three intended symbols ($a_1$, $a_2$ and $a_3$ for $\r{1}$), and three interfering symbols ($b_1+c_1$, $b_2+c_2$, and $b_3+c_3$ for $\r{1}$), which can be solved. Note that the system of linear equations might not be linearly independent, depending of $p$, the field size. In particular, for these specific parameters, operating in the binary field ($p=2$), the coefficient of $a_3$ becomes zero, and therefore $a_3$ cannot be decoded from the received equations. However, $p$ is an arbitrary parameter, which can be carefully chosen to provide a full-rank coefficient matrix. Therefore, a per-user rate of $3/2$ symbols/channel-use is achieved with feedback. \\[-1mm]
\begin{figure*}[t]
\begin{center}
 	\psfrag{t1}[Bc][Bc]{$\t{1}$}
	\psfrag{t2}[Bc][Bc]{$\t{2}$}
	\psfrag{t3}[Bc][Bc]{$\t{3}$}
	\psfrag{r1}[Bc][Bc]{$\r{1}$}
	\psfrag{r2}[Bc][Bc]{$\r{2}$}
	\psfrag{r3}[Bc][Bc]{$\r{3}$}
	\psfrag{a1}[Bc][Bc]{$a_1$}
	\psfrag{a2}[Bc][Bc]{$a_2$}
	\psfrag{a3}[Bc][Bc]{$a_3$}
	\psfrag{a4}[Bc][Bc]{$(b_1+c_1)$}
	\psfrag{a5}[Bc][Bc]{$(b_2+c_2)$}
	\psfrag{a6}[Bc][Bc]{$(b_3+c_3)$}
	\psfrag{b1}[Bc][Bc]{$b_1$}
	\psfrag{b2}[Bc][Bc]{$b_2$}
	\psfrag{b3}[Bc][Bc]{$b_3$}
	\psfrag{b4}[Bc][Bc]{$(a_1+c_1)$}
	\psfrag{b5}[Bc][Bc]{$(a_2+c_2)$}
	\psfrag{b6}[Bc][Bc]{$(a_3+c_3)$}
	\psfrag{c1}[Bc][Bc]{$c_{1}$}
	\psfrag{c2}[Bc][Bc]{$c_{2}$}
	\psfrag{c3}[Bc][Bc]{$c_{3}$}
	\psfrag{c4}[Bc][Bc]{$(a_1+b_1)$}
	\psfrag{c5}[Bc][Bc]{$(a_2+b_2)$}
	\psfrag{c6}[Bc][Bc]{$(a_3+b_3)$}
	\psfrag{T1}[Bc][Bc]{$T=1$}
	\psfrag{T2}[Bc][Bc]{$T=2$}
	\psfrag{aa1}[Bc][Bc]{$(b_1+c_1)$}
	\psfrag{aa2}[Bc][Bc]{$a_1+(b_2+c_2)$}
	\psfrag{aa3}[Bc][Bc]{$a_2+(b_3+c_3)$}
	\psfrag{aa4}[Bc][Bc]{$2a_1+(b_1+c_1)$}
	\psfrag{aa5}[Bc][Bc]{$2a_{2}+(b_2+c_2)$}
	\psfrag{aa6}[Bc][Bc]{$2a_3+(b_3+c_3)+(b_1+c_1)$}
	\psfrag{bb1}[Bc][Bc]{$(a_1+c_1)$}
	\psfrag{bb2}[Bc][Bc]{$b_1+(a_2+c_2)$}
	\psfrag{bb3}[Bc][Bc]{$b_2+(a_3+c_3)$}
	\psfrag{bb4}[Bc][Bc]{$2b_1+(a_1+c_1)$}
	\psfrag{bb5}[Bc][Bc]{$2b_{2}+(a_2+c_2)$}
	\psfrag{bb6}[Bc][Bc]{$2b_3+(a_3+c_3)+(a_1+c_1)$}
	\psfrag{cc1}[Bc][Bc]{$(a_1+b_1)$}
	\psfrag{cc2}[Bc][Bc]{$c_1+(a_2+b_2)$}
	\psfrag{cc3}[Bc][Bc]{$c_2+(a_3+b_3)$}
	\psfrag{cc4}[Bc][Bc]{$2c_1+(a_1+b_1)$}
	\psfrag{cc5}[Bc][Bc]{$2c_{2}+(a_2+b_2)$}
	\psfrag{cc6}[Bc][Bc]{$2c_3+(a_3+b_3)+(a_1+b_1)$}
\includegraphics[width=0.85\textwidth]{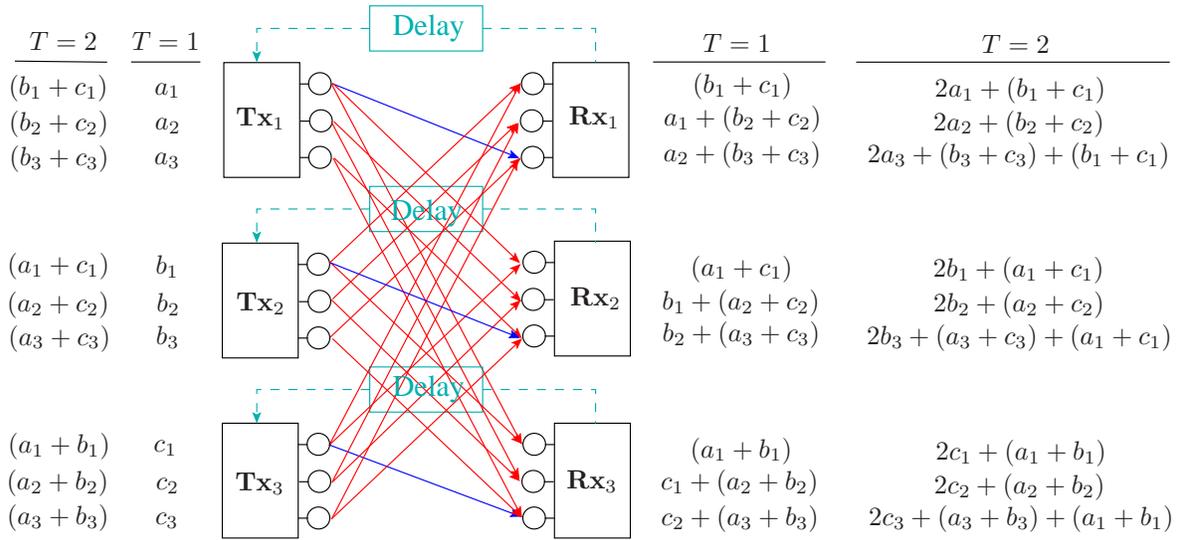}
\end{center}
\caption{Coding scheme for the linear deterministic model in the strong interference regime, for $K=3$, $n=1$, and $m=3$. }
\label{fig:det-strong}
\end{figure*}

\paragraph{Moderate Interference Regime $(m = n)$}
As discussed in the outer bound argument, the capacity curve is discontinuous at $m=n$. A trivial encoding scheme to achieve rate $\Rs=n/K$ is to perform time-sharing over $K$ blocks: in block $k$ only $\t{k}$ transmits its message at rate $R_k=n$ while all the transmitters keep silent. Note that this coding scheme does not get any benefit from the feedback link. 

\subsection{Outer Bound}
\label{sec:det-OB}
In this section we derive an outer bound on the symmetric feedback capacity of the fully-connected interference channel. We may use shorthand notation $W_{[2:K]}$ to denote $(W_2,W_3,\dots, W_K$. Similarly $Y_{[2:K]t}$ may be used to denote $(Y_{2t},Y_{3t},\dots,Y_{Kt})$.

Assume there exists an encoding scheme with block length $T$, which can reliably convey messages of each transmitter to its intended receiver.  We begin with the following chain of inequalities:
\begin{align}
H(W_1)&+ H(W_2) \stackrel{(a)}{=} H(W_1,W_2 | W_{[3:K]} ) \nonumber\\
&\leq  H(W_1,W_2, Y_1^T,Y_2^T | W_{[3:K]})\nonumber\\
&= H(Y_2^T | W_{[3:K]} ) + H(W_2  | W_{[3:K]}, Y_2^T) \nonumber\\
&\phantom{=}+ H(Y_1^T  | W_{[3:K]},Y_2^T) +  H(W_1  | W_{[3:K]} ,Y_1^T,Y_2^T)\nonumber\\
&\leq H(Y_2^T) + H(W_2| Y_2^T) \nonumber\\
&\phantom{=}+ H(Y_1^T  | W_{[3:K]} ,Y_2^T) +H(W_1|Y_1^T)\nonumber\\
&\leq T[\max(m,n)+2\epsilon_T] +H(Y_1^T  | W_{[3:K]} ,Y_2^T),\label{eq:fano-1}
\end{align}
where $(a)$ holds since messages are assumed to be independent, and \eqref{eq:fano-1} is due to Fano's inequality, in which $\epsilon_T\rightarrow 0 $, as $T$ grows. We can continue with bounding the remaining term in \eqref{eq:fano-1} as 
\begin{align}
H(Y_1^T  &| W_{[2:K]},Y_2^T) \nonumber\\
&\leq H(Y_1^T, Y_{[3:K]}^T,  | W_{[2:K]}Y_2^T)\nonumber\\
&=\sum_{t=1}^T H(Y_{1t}, Y_{[3:K]t} | W_{[2:K]}, Y_2^T, Y_1^{t-1}, Y_{[3:K]}^{t-1})\nonumber\\
&\stackrel{(b)}{=}\sum_{t=1}^T H(Y_{1t}, Y_{[3:K]t} | W_{[2:K]}, Y_2^T, Y_1^{t-1}, Y_{[3:K]}^{t-1}, X_{[2:K]t} )\nonumber\\
&\stackrel{(c)}{\leq} \sum_{t=1}^T H(Y_{1t}, Y_{[3:K]t}  |  Y_2^T,  X_{[2:K]t})\nonumber\\
&\stackrel{(d)}{=}\sum_{t=1}^T H(Y_{1t}, Y_{[3:K]t} | Y_2^T, X_{[2:K]t}, D^{q-m} X_{1t}  )\nonumber\\
&\stackrel{(e)}{=} \sum_{t=1}^T H(D^{q-n} X_{1t}   | Y_2^T, X_{[2:K]t}, D^{q-m} X_{1t} )\nonumber\\
&\stackrel{(c)}{\leq}  \sum_{t=1}^T H(D^{q-n} X_{1t}   | D^{q-m} X_{1t})\nonumber\\
&=T(n-m)^+,\label{eq:term4}
\end{align}
where $(b)$ is due to the fact that $X_{jt}=f_{jt}(W_j, Y_j^{t-1})$; $(c)$ holds because conditioning reduces entropy; $(d)$ follows the fact that  $D^{q-m} X_{1t}= Y_{2t}-D^{q-n} X_{2t} -D^{q-m} \sum_{j>2} X_{jt} $ is a deterministic function of $(Y_{2t}, X_{[2:K]t})$;  and  
$(e)$ holds because given $(D^{q-m} X_{1t}, X_{[2:K]t})$, the output $Y_{jt}=D^{q-n} X_{jt} +D^{q-m} \sum_{i\neq j} X_{it}  $ is  deterministically known  for $j=3,\dots, K$; moreover, every term in $Y_{1t}$ except $D^{q-n} X_{1t}$ is know given the same condition. 

Replacing \eqref{eq:term4} in \eqref{eq:fano-1} we  arrive at 
\begin{align}
R_1+R_2 &\leq \frac{1}{T} [H(W_1)+H(W_2)] \nonumber\\
&\leq \max(m,n) + (n-m)^+ +2\epsilon_T \nonumber\\
&=\max (m,2n-m) +2\epsilon_T.
\end{align}
Finally, since we are interested in symmetric rate characterization, we can set $R_1=R_2$, which yields 
\begin{align}
\Rs\leq \max\left(\frac{m}{2}, n-\frac{m}{2}\right) +\epsilon_T.
\label{eq:Rsym-m-neq-n}
\end{align}
Letting $T\rightarrow\infty$ and $\epsilon_T\rightarrow 0$, we obtain the upper bound as claimed in Theorem~\ref{thm:det}. 

The capacity behavior of the network has a discontinuity at $m=n$, where the symmetric achievable rate scales inverse linearly with $K$. The reason behind this phenomenon is very apparent by focusing on the deterministic model. This study reveals that when $m=n$ the received signals at all the receivers are \emph{exactly} the same. Therefore, each receiver should be able to decode all the messages, and hence its decoding capability is shared between all the signals, which results in $\Rs=n/K$. More formally, we can write
\begin{align}
T\sum_{k=1}^K R_k &= H(W_1,W_2,\dots, W_K) \nonumber\\
&\leq  I(W_{[1:K]}; Y_{[1:K]}^T) + KT\epsilon\nonumber\\
&\stackrel{(f)}{=}  I(W_{[1:K]}; Y_1^T) + KT\epsilon\nonumber\\
&\leq H(Y_1^T) + KT\epsilon 
\leq Tn + KT\epsilon,\label{eq:UB-n-equal-m}
\end{align} 
where $(f)$ is due to the fact that $Y_1^T=Y_2^T=\cdots=Y_K^T$. Dividing \eqref{eq:UB-n-equal-m} by $KT$ and setting $R_1=\dots=R_K=\Rs$, we arrive at $\Rs\leq n/K$.

\section{The Gaussian Network: A Coding Scheme}
\label{sec:G-Ach}
The encoding scheme we propose for this problem is similar to that of the $2$-user case. It is shown in \cite{SuhTse11} that for the $2$-user feedback interference channel, depending on the interference regime (value of $\alpha$), it is (approximately) optimum to decode the interfering message. Due to existence of the feedback, decoding the interference is not only useful for its removal and consequent decoding of the desired message (akin to the strong interference regime without feedback), but also helps for decoding  a part of the intended  message that is conveyed through the feedback path. In the $2$-user case,  at the end of the transmission block, each receiver not only decodes its own message completely, but also partially decodes the message of the other receiver.

A fundamental difference here is that in the $K$-user problem, there are multiple interfering messages that can be heard  at each receiver. Partial decoding of all interfering messages would dramatically decrease the maximum rate of the desired message. Our approach to deal with this is to consider the total interference received from all other users as a single message and decode it, without resorting to resolving the individual component of the interference. There are two key conditions to be fulfilled that allow us to perform such decoding, namely, {\sf{(i)}} interfering signals should be  \emph{aligned}, and {\sf{(ii)}} the summation of interfering signals should belong to a message set of proper size which can be decoded at each receiver. Here, the first condition is satisfied since the network is symmetric (all the interfering links have the same gain), and therefore all the interfering messages are received at the same power level. In order to satisfy the second condition, we can use a common \emph{lattice code} in all transmitters, instead of random Gaussian codebooks. The structure of a lattice codebook and its closedness with respect to summation, imply that the summation of aligned interfering codewords observed at each receiver is still a codeword from the same codebook. This allows us to perform decoding by searching over the single codebook, instead of the Cartesian product of all codebooks. Due to the fact that the aligned interference is decoded, we call this coding scheme \emph{cooperative interference alignment}. 

\paragraph*{Lattice Codes}

\emph{Lattice code} is a class of codes that can achieve the capacity of the Gaussian channel \cite{urbanke-lattice, erez2004achieving}, with lower complexity compared to the conventional random codes. The structural behaviors of lattice codes is very important property which can also be exploited for interference alignment. 

In the following we present a brief introduction for lattice codes which will be used later in our coding strategy. 

A $T$-dimensional lattice $\Lambda$ is subset of $T$-tuples with real elements, such that $\mathbf{x},\mathbf{y}\in \Lambda$ implies $-\mathbf{x}\in \Lambda $ and $\bx+\by\in\Lambda$. For an arbitrary $\bx \in \mathbb{R}^T$, we define $[\bx \mod \Lambda]= \bx - Q(\bx)$, where
\[
Q(\bx)=\arg \min_{\mathbf{t} \in \Lambda} \parallel \bx-\mathbf{t}\parallel 
\]
is the closet lattice point to $\bx$. The Voronoi cell of $\Lambda$ denoted by $\cV$ is defined as
\[
\cV=\{\bx\in\mathbb{R}^T: Q(\bx)=\mathbf{0}\}.
\]
The Voronoi volume $V(\cV)$ and the second moment $\sigma^2(\Lambda)$ of the lattice are defined as
\begin{align*}
V(\cV)=\int_{\cV} d\bx, \qquad \qquad
 \sigma^2(\Lambda)=\frac{ \int_{\cV} \parallel \bx\parallel^2 d\bx}{T V(\cV)}.
\end{align*}
We further define the normalized second moment of $\Lambda$ as
\[
G(\Lambda) = \frac{\sigma^2(\Lambda)}{V(\cV)^{2/T}}=\frac{1}{T} \frac{\int_{\cV} \parallel \bx\parallel^2 d\bx }{ V(\cV)^{1+\frac{2}{T} }}.
\]

A sequence of lattices $\{\Lambda_T\}$ is called \emph{good quantization code} if 
\[
\lim_{T\rightarrow\infty} G(\Lambda_T) =\frac{1}{2\pi e}.
\]
On the other hand a sequence of lattices is known to be \emph{good for AWGN channel coding} if 
\[
\lim_{T\rightarrow\infty} \Pr[\bz^T \in \cV_T]=1,
\]
where $\bz^T\sim\mathcal{N}\left(0, \sigma^2(\Lambda_T)\right)$ is random zero-mean Gaussian noise with proper variance. It is shown in \cite{ELZ-lattices} that there exist sequences of lattices $\{\Lambda_T\}$ which are simultaneously good for quantization and AWGN channel coding. 

In the rest of this section, we prove the direct part of Theorem~\ref{thm:Cap-Gaussian}. The analysis of two cases, namely weak and strong interference regimes, is separately presented. 

\subsection{Weak Interference Regime $2 \leq \INR \leq  \SNR/2$}
\label{ssec:GVW}
The coding scheme we use for this regime is based on the insight gained from studying the deterministic model. A careful review of the coding scheme illustrated in Appendix~\ref{app:det-ach-weak} reveals that the set of information symbols of each user can be split into three subsets: {\sf (1)} $(S_k(1),\dots, S_k(m))$ that are sent over the first channel use and cause interference for other receivers; {\sf (2)}   $(S_1(m+1),\dots, S_k(n))$ which are corrupted by interference at $\r{k}$, but do not cause interference at other receivers; and finally {\sf (3)} $(S_1(n+1),\dots, S_k(2n-m))$ which are sent on the second channel use on proper levels such that they do not cause interference at other receivers. The other $m$ levels of each transmitter in the second channel use send the interfering signal received at $\r{k}$ in the previous channel use. In the decoding process, each receiver first decodes the \emph{total} interference from its channel output in the second channel use, and removes it to decode $(S_1(n+1),\dots, S_k(2n-m))$. Then it also subtract the interference from its channel output in the first slot in order to decode $(S_1(1),\dots, S_k(m))$ and $(S_1(m+1),\dots, S_k(n))$.

Inspired by the this coding scheme and message splitting, we  consider three messages $w_{k0}$, $ w_{k1}$, and $w_{k2}$, for transmitter $\t{k}$ which will be conveyed to  receiver $\r{k}$ over two blocks. All similar sub-messages from different users have the same rates, which are denoted by $R_{k0}$, $R_{k1}$, and $R_{k2}$.  Encoding of $w_{k1}$ and $w_{k2}$ (which are counterparts of $(S_k(m+1),\dots, S_k(n))$ and $(S_k(n+1),\dots,S_{k}(2n-m))$, respectively) is performed  using usual random Gaussian codebooks with block length $T$ and average power $1$, which results in codewords $\bc_{k1}$ and $\bc_{k2}$. 
The power allocated to $\bc_{k1}$ and $\bc_{k2}$ is chosen such that they get received at other receivers at the noise level. 

The third sub-message, $w_{k0}$ (corresponding to $(S_k(1),\dots, S_k(m))$) is the main interfering part from $\t{k}$. Since we need the total interference to be decodable, we need to use a common lattice code which is shared between all transmitters. 

We need a nested lattice code \cite{ZSE-lattice} which is generated using a good quantization lattice for shaping and a good channel coding lattice. We start with $T$-dimensional nested lattices $\Lambda_c\subseteq \Lambda_f$, where $\Lambda_c$ is a good quantization lattice $\Lambda_c$ with $\sigma^2(\Lambda_c)=1$ and $G(\Lambda_c)\approx 1/2\pi e$, and $\Lambda_f$ as a good channel coding lattice. We  construct a codebook $\cC=\Lambda_f \cap \cV_c$, where  $\cV_c$ is the Voronoi cell of the  lattice $\Lambda_c$. The following properties are fairly standard in the context of lattice coding: 
\begin{itemize}
\item[a)] Codebook $\cC$ is a closed set with respect to summation under the ``$\mod \Lambda_c$'' operation, i.e., if $\bx_1,\bx_2\in\cC$ are two codewords, then $[\bx_1+\bx_2] \mod \Lambda_c \in\cC$ is also a codeword. 
\item[b)] Lattice code $\cC$  can be used to reliably transmit up to rate\footnote{A more sophisticated scheme can achieve rates $R=\frac{1}{2}\log \left(1+ \SNR\right)$. However, the simple scheme is sufficient for the purpose of approximate capacity characterization.} $R=\frac{1}{2}\log (\SNR)$  over a Gaussian channel modelled by $Y=\sqrt{\SNR}X+Z$ with $\E[Z^2]=1$.
\end{itemize}

In order to encode $w_{k0}$, we use the common lattice code $\cC$ defined above.  Let $\bs_{k0}$ be the lattice codeword to which $w_{k0}$ is mapped, and define $\bs_0=\bs_{10}+\bs_{20}+\dots+\bs_{K0}$. 

Once the encoding process is performed, the signal transmitted by $\t{k}$ in the first block (of length $T$) is formed as
\begin{align}
\bx_{k1}=\sqrt{\frac{\INR-1}{\INR}} \bc_{k0} + \sqrt{\frac{1}{\INR}} \bc_{k1}.\nonumber
\end{align}
where $\bc_{k0}=[\bs_{k0} -\bd_k] \mod \Lambda_c$, and $\bd_k$ is a random dither uniformly distributed over $\cV_c$, and shared between all the terminals in the network. Therefore, the signal received at $\r{k}$ can be written as
\begin{align}
\by_{k1}&=\sqrt{\SNR} \bx_{k1} + \sqrt{\INR} \sum_{i\neq k} \bx_{i1} +\bz_{k1}\nonumber\\
&=\sqrt{\frac{\SNR}{\INR} (\INR-1)} \bc_{k0} + \sqrt{\frac{\SNR}{\INR}} \bc_{k1} 
+ \sqrt{\INR-1} \sum_{i\neq k} \bc_{i0} \nonumber\\
&\phantom{=}+ \sum_{i\neq k} \bc_{i1} +\bz_{k1}.\nonumber
\end{align}
This received signal is sent to the transmitter $\t{k}$ over the feedback link. Having $\bx_{k1}$ and $\by_{k1}$, the transmitter  can compute 
\begin{align}
\tilde{\by}_{k}&=\by_{k1}-(\sqrt{\SNR}-\sqrt{\INR})\bx_{k1}= \sqrt{\INR} \sum_{i=1}^K \bx_{i1}+\bz_{k1}\nonumber\\
&=  \sqrt{\INR-1} \sum_{i=1}^K \bc_{i0} + \sum_{i=1}^K \bc_{i1} +\bz_{k1}.\nonumber
\end{align}
Recall that $\bs_0=[\sum \bs_{i0} \mod \Lambda_c] = [\sum \bc_{i0} + \sum \bd_i \mod \Lambda_c]  \in\cC$. So it can be decoded  from $\tilde{\by}_k$  by treating the rest as noise, provided that  
\begin{align}
R_0 \leq \frac{1}{2} \log \left(  \frac{\INR-1}{K+1}\right). \label{eq:W-R0-1}
\end{align}
Note that at this point $\r{k}$ cannot decode $\bc_0$. 

In the second block, having $\bs_0$ decoded,  $\t{k}$  generates $\bc_0=[\bs_0 - \bd_0] \mod \Lambda_c$ and transmits
\begin{align*}
\bx_{k2}=\sqrt{\frac{\INR-1}{\INR}} \bc_{0} + \sqrt{\frac{1}{\INR}} \bc_{k2}.
\end{align*} 
The signal received at $\r{k}$ in the second block can be written as
\begin{align}
\by_{k2}&=\sqrt{\SNR} \bx_{k2} + \sqrt{\INR} \sum_{i\neq k} \bx_{i2} +\bz_{k2}\\
&=\sqrt{\frac{\SNR}{\INR} (\INR-1)} \bc_{0} + \sqrt{\frac{\SNR}{\INR}} \bc_{k2} + \sqrt{\INR-1} \sum_{i\neq k} \bc_{0} \nonumber\\
&\phantom{=}+ \sum_{i\neq k} \bs_{i2} +\bz_{k2}\nonumber\\
&=\left( \sqrt{\frac{\SNR}{\INR}} +K-1\right)\sqrt{\INR-1} \bc_0 + \sqrt{\frac{\SNR}{\INR}} \bc_{k2} \nonumber\\
&\phantom{=}+ \sum_{i\neq k} \bc_{i2} +\bz_{k2}.
\end{align}
Receiver $\r{k}$ first decodes $\bc_{0}$ treating everything else as noise. This is possible as long as 
\begin{align}
R_0 \leq \hspace{-1pt}\frac{1}{2}\log\hspace{-1pt}\left(\hspace{-1pt}\frac{(\INR\hspace{-1pt}-\hspace{-1pt}1)\hspace{-1pt}\left(\hspace{-1pt}\sqrt{\SNR} +(K\hspace{-1pt}-\hspace{-1pt}1)\sqrt{\INR}\right)\hspace{-1pt}^2}{\SNR+K\INR}\hspace{-1pt}\right)\hspace{-1pt}.\label{eq:W-R0-2}
\end{align}
After decoding and removing $\bc_0$ from the received signal, $\r{k}$ can decode the Gaussian codeword $\bc_{k2}$, provided that
\begin{align}
R_2 \leq \frac{1}{2} \log \left( 1+ \frac{\SNR}{K \INR}\right). \label{eq:W-R2-1}
\end{align}
Next, the decoder uses $\bc_0$ to remove the interference $\sum_{i\neq k} \bc_{0i}$ from $\by_{k1}$ in order to  decode $\bc_{k0}$ and $\bc_{k1}$. To this end, $\r{k}$ first computes 
\begin{align}
\by_{k1}-\sqrt{\INR-1}\bc_0=&
\left(\sqrt{\frac{\SNR}{\INR} }-1\right) \sqrt{\INR-1} \bc_{k0}  \nonumber\\
&+   \sqrt{\frac{\SNR}{\INR}} \bc_{k1}+ \sum_{i\neq k} \bc_{i1} +\bz_{k1},\nonumber
\end{align}
from which  codewords $\bc_{k0}$ and $\bc_{k1}$ can be sequentially decoded provided that
\begin{align}
R_0 &\leq \frac{1}{2}\log\left(\frac{(\INR-1)\left(\sqrt{\SNR} -\sqrt{\INR}\right)^2}{\SNR+K \INR}\right),\label{eq:W-R0-3}\\
R_1 &\leq \frac{1}{2} \log \left( 1+ \frac{\SNR}{K \INR}\right). \label{eq:W-R1-1}
\end{align}
It only remains to choose $R_0$, $R_1$, and $R_2$ that satisfy all constraints in \eqref{eq:W-R0-1}--\eqref{eq:W-R1-1}. It is easy to verify that the choice of
\begin{align}
\begin{split}
R_0^\star&=\frac{1}{2}\log\left(\frac{\INR-1}{8(K+1)} \right)\\ 
R_1^\star&=R_2^\star=\frac{1}{2} \log \left( 1+ \frac{\SNR}{K\INR}  \right) \label{eq:rate-for-GW}
\end{split}
\end{align}
satisfies all the constraints, and  therefore 
\begin{align*}
\Rs&=\frac{1}{2}(R_0^\star+R_1^\star+R_2^\star)\\
&=\frac{1}{4}\log\left(\frac{\INR-1}{8(K+1)}  \right)+\frac{1}{2} \log \left( 1+ \frac{\SNR}{K \INR}\right) 
\end{align*}
can be simultaneously achieved for all the $K$ pairs of transmitters/receivers.

In the following we rephrase this achievable rate in a manner so that it can be easily compared to $\Cs$ in Theorem~\ref{thm:Cap-Gaussian}.   It is easy to verify that for $2\leq \INR \leq \frac{1}{2}\SNR$ we have
\begin{align}
\frac{\INR-1}{8(K+1)}\left(1+\frac{\SNR}{K\INR}\right) \geq \frac{1}{16K(K+1)} \left(1+\INR +\SNR \right),
\label{eq:GW-simplify-1}
\end{align}
which implies 
\begin{align}
\frac{1}{4}\log &\left(\frac{\INR-1}{8(K+1)}\right) + \cg{\frac{\SNR}{K\INR}} \nonumber\\
&\geq \cgc{1}{4}{\INR+\SNR} +\cgc{1}{4}{\frac{\SNR}{1+\INR}}\nonumber\\
&\phantom{=} -\frac{1}{4} \log 16 K^2(K+1)
\label{eq:Rw-common-form2}
\end{align}
Therefore, for this regime the symmetric rate of
\begin{align}
\Rs=& \cgc{1}{4}{\INR+\SNR} +\cgc{1}{4}{\frac{\SNR}{\INR}} \nonumber\\
&-\frac{1}{4} \log 16 K^2(K+1)
\end{align}
is achievable. 

\begin{remark}
It is worth mentioning that the coding schemes proposed for the weak interference regimes keep all messages except $W_k$ \emph{almost} secure from receiver $\r{k}$, for all $k=1,\dots,K$. More precisely, one can show that for $K\geq 3$, the leakage rate of information is upper bounded by 
\begin{align}
\frac{1}{2T} I(W_k ; y_{j}^{2T} ) \leq \frac{1}{2}\log \frac{K}{K-1}, \qquad k\neq j,
\end{align} 
where $2T$ is the length of the entire course of communication. Here the upper bound on the leakage rate is a constant, independent of  $\SNR$, $\INR$, and the actual rates of the messages. However, this secrecy is different from (and weaker than) the standard notion of secrecy, which imposes a vanishing total leakage rate in strong secrecy, or a vanishing per-symbol leakage rate in weak secrecy\footnote{We refer the reader to \cite{liang2009information} (and references therein) for details concerning information-theoretic secrecy. It is worth mentioning that both weak and strong secrecy are shown to be equivalent in  \cite{maurer2000information}, in the sense that  substituting the weak secrecy criterion by the stronger version does not change the secrecy capacity.}. 

The main intuition behind this is the following: each receiver can only decode its own message, as well as the sum-lattice codeword corresponding to the messages of other users. For instance, after decoding $W_1$, $\r{1}$ remains with a codeword that depends on $W_2, W_3,\dots,W_K$. Hence, $W_3, \dots,W_K$ act as a mask (encryption key) to hide $W_2$ from $\r{1}$. Therefore, although $\r{1}$ receives a certain amount of information about a function of all other messages, the amount of information it gets about each unintended individual message is negligible. This phenomenon is very similar to the encoding scheme used in  \cite{MDPS:11} to guarantee information-secrecy. However,  here this secrecy is naturally provided by the coding scheme, without any additional penalty in terms of the symmetric achievable rate of the network. We will discuss this property of the encoding scheme in more detail in Appendix~\ref{app:secrecy}.
\label{rmk:secrecy} 
\end{remark}

\subsection{Strong Interference Regime $\INR \geq 2\max(\SNR,1)$}
\label{ssec:GS}
The coding scheme for the strong interference regime is simpler than the last case. It is known that for strong interference regime in the usual interference channel (without feedback) it is  optimum to decode the interference and remove it from the received signal before decoding the intended message \cite{HK:81}. Surprisingly, this is not the case when transmitters get feedback from their respective receivers (as far as  approximate capacity is concerned). In this regime, the receivers do not need to decode the interference, and can cancel it using a zero-forcing scheme. This is implemented using Alamouti's scheme \cite{alamouti1998simple} in \cite{SuhTse11} for $K=2$. The orthogonality of the design matrix in the $2\times 2$ Alamouti's scheme causes the intended signal and the interference signal to be orthogonal, and so zero-forcing the interference does not cause a loss in signal power. However, it is shown in \cite{tarokh1999space} that $K\times K$ orthogonal designs  exist only for $K=2,4,8$ with real elements, and $K=2$ with complex elements.  
When such  matrices do not exist for arbitrary $K$, we may use non-orthogonal coding for the intended and interference signal. The key point is that if the transmitters can re-generate the interfering signal of one coding block at the receivers over another block, then the receiver can cancel two copies of interference, and decode its intended message. Since the transition matrix between the intended and interference signal on one side and the channel outputs on the other side is non-orthogonal, zero-forcing causes a power loss. However, this only affects the gap between the achievable rate and the upper bound, and does not cause a major problem when approximate capacity is concerned. We present this scheme in general detail in the rest of this section.

As in the previous case, the transmission is performed over two blocks. First note that since the interfering signals do not need to be decoded neither at the transmitters nor receivers, there is no need to use lattice codes to force the total interference to be a codeword. We associate a randomly generated Gaussian codebook of rate $\Rs$ to each transmitter. Transmitter $\t{k}$ maps its message $w_k$ to a codeword $\bc_k$ from its codebook, and sends $\bx_{k1}=\bc_k$ over the first block of transmission. At the end of the first block, each receiver sends back its received signal to its respective transmitter. Upon receiving $\by_{k1}$ from the feedback link, $\t{k}$ removes its own signal, and resends the residual over the second block. 
\begin{align*}
\bx_{k2}&=\gamma \left[\by_{k1}+ (\sqrt{\INR} -\sqrt{\SNR} ) \bx_{k1}\right] \nonumber\\
&=  \gamma (\sqrt{\INR} \sum_{i=1}^K \bc_{i} +\bz_{k1}),
\end{align*}
where $\gamma=1/\sqrt{K\INR +1}$ guarantees the transmit signal satisfies the power constraint.

At the end of the second transmission block, $\r{k}$ has access to
\begin{align}
\by_{k2}&=\sqrt{\SNR} \bx_{k2} +\sqrt{\INR} \sum_{i\neq k} \bx_{i2} +\bz_{k2}\nonumber\\
&=\gamma \left( \sqrt{\SNR} +(K-1) \sqrt{\INR} \right) \sqrt{\INR} \left(\bc_{k} + \sum_{i\neq k} \bc_{i} \right) \nonumber\\
&\phantom{=}+\gamma\sqrt{\SNR} \bz_{k1} + \gamma \sqrt{\INR} \sum_{i\neq k} \bz_{i1} + \bz_{k2}.\nonumber
\end{align}
Applying zero-forcing at $\r{k}$ to remove $\sum_{i\neq k} \bc_i$, we obtain an effective channel 
\begin{align*}
\tilde{\by}_k&= \by_{k2}- \gamma \left( \sqrt{\SNR} +(K-1) \sqrt{\INR} \right)  \by_{k1}
\nonumber\\
&=\gamma  \left( \sqrt{\SNR} +(K-1) \sqrt{\INR} \right) ( \sqrt{\INR} -  \sqrt{\SNR}) \bc_k\nonumber\\
&\phantom{=}- \gamma (K-1) \sqrt{\INR} \bz_{k1} + \gamma \sqrt{\INR} \sum_{i\neq k} \bz_{i1} + \bz_{k2}\nonumber\\
&= \gamma \left( \sqrt{\SNR} +(K-1) \sqrt{\INR} \right) ( \sqrt{\INR} -  \sqrt{\SNR}) \bc_k + \tilde{\bz}_{k2}
\end{align*}
which can be used for decoding $\bc_k$. The power of the total noise in this effective channel would be 
\begin{align*}
\E[\tilde{\bz}_{k2}^2]&= \gamma^2(K-1)^2 \INR\  \E[\bz_{k1} ^2] + \gamma^2 \INR \sum_{i\neq k} \E[\bz_{i1} ^2] + \E[\bz_{k2} 2] \nonumber\\
&= \frac{\INR}{K\INR+1} \left[ (K-1)^2 + (K-1) \right]+1 \nonumber\\
&= \frac{K^2 \INR +1}{K\INR+1} <K.
\end{align*}
On the other hand the power of the signal in the effective channel can be lower bounded by
\begin{align*}
\gamma^2 \Big( \sqrt{\SNR} +(&K-1) \sqrt{\INR} \Big)^2 ( \sqrt{\INR} -  \sqrt{\SNR}) ^2 \nonumber\\
&\geq
\gamma^2 \left( \sqrt{\SNR} + \sqrt{\INR} \right)^2 ( \sqrt{\INR} -  \sqrt{\SNR}) ^2 \nonumber\\
&=\frac{(\INR-\SNR)^2}{K\INR+1}.
\end{align*}
Therefore, since the course of communication is performed over two blocks, the symmetric rate
\begin{align*}
\Rs=\frac{1}{4} \log\left(1+ \frac{(\INR-\SNR)^2}{K(K\INR+1)}\right)
\end{align*}
can be simultaneously achieved for all pair of transmitter/receiver. We can simplify this expression to make it comparable to the rate claimed in Theorem~\ref{thm:Cap-Gaussian}. First note that 
\begin{align}
1+ \frac{(\INR-\SNR)^2}{K(K\INR+1) }\geq \frac{1}{8K^2} (1+\SNR+\INR)
\label{eq:GS-simplify-1}
\end{align}
for $\INR \geq 2\SNR$ and $K\geq 2$. On the other hand, in this regime we have 
$1+\frac{\SNR}{1+\INR} <2$, which implies
\begin{align*}
\Rs &\geq \frac{1}{4} \log\left(1+\SNR+\INR \right) + \frac{1}{4} \log\left(1+\frac{\SNR}{1+\INR}\right) \\
&\phantom{=}-\frac{1}{4}\log 16K^2.
\end{align*}

\subsection{Negligible Interference Regime $\INR<2$}
\label{ssec:GN}
In the discussion of Sections~\ref{ssec:GVW} and \ref{ssec:GS} we excluded the cases where $\INR$ is small. If this is the case, the standard \emph{treat interference as noise} scheme is close to be optimum. Here we briefly discuss the achievable rate and its gap from the upper bound for completeness. 

In this regime, each transmitter encodes its message using a Gaussian codebook, and sends it to the receiver. The course of communication is performed in a single block, and each receiver decodes its message at the end of the block by treating the interference as noise. This can support any positive rate not exceeding
\begin{align*}
\Rs=\frac{1}{2}\log\left(1+\frac{\SNR}{1+(K-1)\INR}\right).
\end{align*}
This expression can be rephrased as
\begin{align*}
\Rs&>\frac{1}{2}\log\frac{1}{K-1}\left(1+\frac{\SNR}{1+\INR}\right)\nonumber\\
&= \frac{1}{4}\log\left(1+\SNR+\INR\right)  -\frac{1}{4}\log(1+\INR)\nonumber\\
&\phantom{=} + \frac{1}{4} \log\left( 1+ \frac{\SNR}{\INR+1}\right) -\frac{1}{2} \log (K-1)\nonumber\\
&\geq \Cs- \frac{1}{4}\log 3(K-1)^2,
\end{align*}
which implies $\Cs-\Rs \leq \frac{1}{4}\log 3(K-1)^2$.

\section{The Gaussian Network: An Upper Bound}
\label{sec:G-UB}
In this section we prove the converse part of Theorem~\ref{thm:Cap-Gaussian}. To this end, we derive an upper bounds on the symmetric rate of the network. The essence of this bound is the same as the converse proof  for the deterministic network. That is, in the strong interference regime, given all the messages except for two of them, the output signal of any of the respective receivers is not only sufficient to decode its own message, but can also be used to decode the other missing message. Similarly, in the weak interference regime, although one receiver cannot completely decode the message of the other transmitter, it receives enough information to partially decode that message. 

We first define $\tz_{it}=z_{it}-z_{2t}$ for $i=3,4,\dots,K$ and $t=1,\dots,T$. Then, we can write 
\begin{align}
T(&R_1+R_2) \leq H(W_1)+H(W_2) \nonumber\\
&\stackrel{(a)}{=} H \Big( W_1,W_2 \Big| W_{[3:K]} \Big) \nonumber\\
&=H\Big(W_2\Big| W_3,\dots,W_K\Big) +H\Big(W_1 \Big| W_{[2:K]}\Big)\nonumber\\
&= I\Big(W_2; y_2^T \Big| W_{[3:K]}\Big) +H\Big( W_2 \Big| y_2^T ,W_{[3:K]}\Big) \nonumber\\
&\phantom{=} +I\Big(W_1; y_1^T y_2^T \Big| W_{[2:K]}\Big) + H\Big( W_1\Big|  y_1^T y_2^T,W_{[2:K]}\Big)\nonumber\\
&\leq I\Big(W_2; y_2^T, \tz_{[3:K]}^T \Big| W_{[3:K]}\Big) \nonumber\\
&\phantom{=}+I\Big(W_1; y_1^T y_2^T, \tz_{[3:K]}^T \Big| W_{[2:K]}\Big)  +2T\epsilon_T\nonumber\\
&= h\Big(y_2^T, \tz_{[3:K]}^T \Big| W_{[3:K]}\Big) - h\Big(y_2^T, \tz_{[3:K]}^T \Big| W_{[2:K]}\Big) \nonumber\\
&\phantom{=}+h\Big(y_1^T y_2^T, \tz_{[3:K]}^T \Big|W_{[2:K]}\Big) - h\Big(y_1^T y_2^T, \tz_{[3:K]}^T \Big| W_{[1:K]}\Big) \nonumber\\
&\phantom{=}+2T\epsilon_T\nonumber\\
&= h\Big(y_2^T, \tz_{[3:K]}^T  \Big| W_{[3:K]}\Big) + h\Big(y_1^T \Big| y_2^T, \tz_{[3:K]}^T , W_{[2:K]}\Big) \nonumber\\
&\phantom{=}- h\Big(y_1^T y_2^T, \tz_{[3:K]}^T \Big| W_{[1:K]}\Big) +2T\epsilon_T,\label{eq:UB-GS-1}
\end{align}
where $\epsilon_T$ vanishes as $T$ grows. Note that we used independence of the messages in $(a)$.
We can bound each term in \eqref{eq:UB-GS-1} individually. The first term can be bounded as 
\begin{align}
h&\Big(y_2^T, \tz_{[3:K]}^T \Big| W_3,\dots,W_K\hspace{-1pt}\Big) \hspace{-2pt}\leq \hspace{-1pt} h\Big(y_2^T\Big) \hspace{-1pt}+\hspace{-1pt} h\Big(\tz_3^T\Big)\hspace{-1pt}+\hspace{-1pt}\dots\hspace{-1pt}+\hspace{-1pt}h\Big(\tz_K^T\Big) \nonumber\\
&\stackrel{(b)}{\leq} T h(y_2) +\frac{T(K-2)}{2} \log (4\pi e)\nonumber\\
&\leq \frac{T}{2} \log \Bigg(1+ \SNR +(K-1) \INR + 2\sqrt{\SNR\cdot \INR}\sum_{j\neq 2}\rho_{2j} \nonumber\\
&\phantom{\leq \frac{T}{2} \log \Bigg(} + 2\INR \sum_{\begin{subarray}\ \ i>j \\ i,j\neq 2\end{subarray}} \rho_{ij} \Bigg) +\frac{T(K-1)}{2} \log (4\pi e)\nonumber\\ 
&\leq \frac{T}{2} \log \hspace{0pt}\left(\hspace{0pt}1\hspace{0pt}+\hspace{0pt}      \left( \sqrt{\SNR}+(K-1)\sqrt{\INR}\right)^2    \right)  \nonumber\\
&+\frac{T(K-1)}{2} \log (4\pi e), \label{eq:term1}
\end{align}
where $\rho_{ij} \in [-1,1]$ is the correlation coefficient between channel inputs $x_i$ and $x_j$. In $(b)$ we used the fact that $\E[\tz_i^2]=2$.

Bounding the second term is more involved. First note that 
\begin{align}
I\Big(y_1^T&; y_{[3:K]}^T \Big| y_2^T, \tz_{[3:K]}^T , W_{[2:K]} \Big)\nonumber\\
&= \sum_{t=1}^T I\Big(y_1^T; y_{[3:K]t} \Big| y_2^T, \tz_{[3:K]}^T , W_{[2:K]}, y_{[3:K]}^{t-1} \Big)\nonumber\\
&\stackrel{(c)}{=} \sum_{t=1}^T I\Big(y_1^T; y_{[3:K]t} \Big| y_2^T, \tz_{[3:K]}^T , W_{[2:K]}, y_{[3:K]}^{t-1}, x_{[2:K]t}\Big)\nonumber \\
&\stackrel{(d)}{=} 0\label{eq:y_i-reconstruct}
\end{align}
where $(c)$  holds since for $j=2,\dots,K$, $x_{jt}=f_{jt}(W_j,y_j^{t-1})$ is a deterministic function of the message and channel output. The  equality in $(d)$ is due to the fact that for $j=3,\dots,K$, we have
\begin{align}
y_{jt}&= \sqrt{\SNR} x_{jt} +\sqrt{\INR} \sum_{i\notin\{2,j\}} x_{it} + \sqrt{\INR} x_{2t}+ z_{jt} \nonumber\\
&= \left[ \sqrt{\SNR} x_{2t} +\sqrt{\INR} \sum_{i\notin \{2,j\}} x_{it} +\sqrt{\INR} x_{jt}+ z_{2t}\right] \nonumber\\
&\phantom{=} +(\sqrt{\SNR}-\sqrt{\INR})(x_{1t}- x_{2t})  +(z_{jt}-z_{2t})\nonumber\\
&=y_{2t} +(\sqrt{\SNR} -\sqrt{\INR}) (x_{jt}-x_{2t}) +\tz_{jt},\label{eq:connect-yk-to-y2}
\end{align}
which implies that $y_{jt}$ can be deterministically recovered from $(y_{2t},x_{2t}, x_{jt}, \tz_{jt})$. Hence, each term in \eqref{eq:y_i-reconstruct} is zero. From \eqref{eq:y_i-reconstruct} we can bound the second term in \eqref{eq:UB-GS-1} as
\begin{align}
h\Big(y_1^T & \Big| y_2^T, \tz_{[3:K]}^T , W_{[2:K]}\Big)  = h\Big(y_1^T \Big| y_{[2:K]}^T , \tz_{[3:K]}^T , W_{[2:K]}\Big)\nonumber\\
&= h\Big(y_1^T \Big| y_{[2:K]}^T , \tz_{[3:K]}^T , W_{[2:K]}, x_{[2:K]}T\Big)\nonumber\\
&\leq h\Bigg(\sqrt{\SNR}x_1^T-\sqrt{\INR}\sum_{i\neq  1} x_{i}^T+ z_{1}^T\Bigg | \nonumber\\
&\phantom{\leq h\Bigg(\ \ \ \ \ \ \ \ \ \ } y_2^T-\sqrt{\SNR}x_2^T-\sqrt{\INR} \sum_{j>2} x_{j}^T,   x_{[2:K]}^T\Bigg)\nonumber\\
&\leq h\Big(\sqrt{\SNR}x_1^T + z_{1}^T \Big|  \sqrt{\INR} x_{1}^T +z_{2}^T\Big)\nonumber\\
&\leq \frac{T}{2} \log \left(1+\frac{\SNR}{1+\INR}\right)+\frac{T}{2} \log (2\pi e).\label{eq:term2}
\end{align}

Finally, we can bound the third term in \eqref{eq:UB-GS-1} as follows:
\begin{align}
h&\Big(y_1^T, y_2^T, \tz_{[3:K]}^T \Big|W_{[1:K]} \Big) \nonumber\\
&= \sum_{t=1}^T h\Big(y_{1t}, y_{2t}, \tz_{[3:K]t} \Big| y_1^{t-1}, y_2^{t-1}, \tz_{[3:K]}^{t-1},  W_{[1:K]}\Big)  \nonumber\\
&\geq \sum_{t=1}^T h\Big(y_{1t}, y_{2t}, \tz_{[3:K]t} \Big| y_1^{t-1}, y_2^{t-1}, \tz_{[3:K]}^{t-1},  W_{[1:K]}, x_{[1:K]t} \Big)  \nonumber\\
&= \sum_{t=1}^T h\Big(z_{1t}, z_{2t}, \tz_{[3:K]t} \Big| y_1^{t-1}, y_2^{t-1}, \tz_{[3:K]}^{t-1} ,  W_{[1:K]}, x_{[1:K]t} \Big)  \nonumber\\
&\stackrel{(e)}{=} \sum_{t=1}^T h\Big(z_{1t}, z_{2t}, \tz_{[3:K]t} \Big)\nonumber\\
&= \sum_{t=1}^T h\Big(z_{[1:K]t}\Big)
=\frac{TK}{2} \log(2\pi e), \label{eq:term3}
\end{align}
where $(e)$ is due to the facts that the channels are memoryless and the noise at time $t$ is independent of all the messages and  signals and noises in the past.  Substituting \eqref{eq:term1}, \eqref{eq:term2} and \eqref{eq:term3} in \eqref{eq:UB-GS-1}, and recalling the fact that we are interested in the maximum $R_1=R_2=\Rs$, we get
\begin{align}
\Rs\leq 
& \frac{1}{4} \log \left(1+ \left(\sqrt{ \SNR} +(K-1)\sqrt{\INR}\right)^2\right) \nonumber\\
&+  \frac{1}{4} \log \left(1+\frac{\SNR}{1+\INR}\right)+\frac{K-1}{4}. \nonumber 
\end{align}
This bound can be further simplified as follows.  It is easy to show that
\begin{align*}
\left( \sqrt{\SNR}+(K-1)\sqrt{\INR}\right)^2 
\leq K^2 (\SNR+\INR)
\end{align*}
which implies
\begin{align}
\Rs &\leq\frac{1}{4} \log \left(1+ K^2 ( \SNR + \INR)\right)+ \frac{1}{4} \log \left(1+\frac{\SNR}{1+\INR}\right)\nonumber\\
&\phantom{=}+\frac{K-1}{4}\nonumber\\
&\leq \frac{1}{4} \log (1+ \SNR + \INR  )+ \frac{1}{4} \log \left(1+\frac{\SNR}{1+\INR}\right) \nonumber\\
&\phantom{=}+\frac{K-1}{4} + \frac{1}{2} \log K,\label{eq:UB-common-form}
\end{align}
which is the desired bound.

\section{The Generalized Degrees of Freedom}
\label{sec:dof}
In this section we prove Theorem~\ref{thm:DoF}. The proof for $\alpha\neq 1$ is straight-forward from Theorem~\ref{thm:Cap-Gaussian} as follows. Recall the achievable symmetric rate in Theorem~\ref{thm:Cap-Gaussian}. Hence, 
\begin{align}
d_{\FB}&(\alpha) = \limsup_{\SNR\rightarrow\infty} \frac{\Rs(\SNR,\alpha)}{\frac{1}{2}\log(\SNR)}\nonumber\\
&\hspace{-5pt}= \limsup_{\SNR\rightarrow\infty} \frac{\frac{1}{4} \log(1+\SNR+\SNR^\alpha)+\frac{1}{4} \log(1+\SNR^{1-\alpha})}{\frac{1}{2}\log(\SNR)}\nonumber\\
&\hspace{-5pt}= \frac{1}{2}\max\{1,\alpha\}+ \frac{(1-\alpha)^+}{2} \nonumber\\
&\hspace{-5pt}=\left\{
\begin{array}{ll}
1-\frac{\alpha}{2} & \alpha<1\nonumber\\
\frac{\alpha}{2} & \alpha>1.
\end{array}
\right.
\end{align}

The concept of generalized degrees of freedom for $\alpha=1$ is more involved, and  a finer look to the problem is necessary. For $\INR=\SNR$ we claim that the degrees of freedom of the network is $1/K$. Note that a simple time-sharing scheme, in which in each block all the transmitters except one are silent,  guarantees a reliable rate of  $\Rs=\frac{1}{2K}\log (1+\SNR)$, which results in $d_{\FB}(\INR=\SNR)\geq 1/K$. 

On the other hand we may use a simple  cut-set argument  in order to show optimality of this $\DoF$ for $\INR=\SNR$. 
Recall that in the deterministic model, the received signal of all the receivers were identical for $m=n$. A similar intuition can explain this phenomenon: when the gain of the direct and cross links are the same, the output signals at all receivers are statistically equivalent,  and given any of them, the uncertainty in the others is small. 
We can formally write
\begin{align}
TK &\Rs = T \sum_{k=1}^K R_k = H\Big(W_{[1:K]}\Big) \nonumber\\
&\leq I \Big(y_{[1:K]}^T ; W_{[1:K]} \Big) + KT\epsilon_T\nonumber\\
&= h\Big(y_{[1:K]}^T\Big) -h\Big( y_{[1:K]}^T \Big| W_{[1:K]} \Big) + KT\epsilon_T\nonumber\\
&= h\Big(y_1^T, z_2^T-z_1^T,\dots, z_K^T-z_1^T\Big) \nonumber\\
&\phantom{=} -\sum_{t=1}^T h\Big( y_{[1:K]t} \Big| y_{[1:K]}^{t-1} , W_{[1:K]}\Big) + KT\epsilon_T\nonumber\\
&\stackrel{(a)}{\leq} h(y_1^T)+ \sum_{k=2}^K h\Big(z_k^T-z_1^T\Big) \nonumber\\
&\phantom{=} - \sum_{t=1}^T h\Big( y_{[1:K]t} \Big| y_{[1:K]}^{t-1}, W_{[1:K]}, x_{[1:K]t} \Big) + KT\epsilon_T\nonumber\\
&= h(y_1^T) +\sum_{k=2}^K h\Big(z_k^T-z_1^T\Big) -\sum_{t=1}^T \sum_{k=1}^K h(z_{kt}) + KT\epsilon_T\nonumber\\
&\leq \frac{T}{2} \log \left( 1+ (\sqrt{\SNR} +(K-1) \sqrt{\INR})^2 \right) \nonumber\\
&\phantom{=} +\frac{(K-1)T}{2} \log 2 + KT\epsilon_T\nonumber\\
&\leq\frac{T}{2} \log \left(1+ K^2 \SNR  \right)  +\frac{(K-1)T}{2}   + KT\epsilon_T,\nonumber
\end{align}
where $(a)$ holds since $x_{kt}=f_{kt} (W_k, y_k^{t-1})$. Dividing by $KT$, we get 
\begin{align}
\Rs\leq \frac{1}{K} \log(1+K^2\SNR)+ \frac{K-1}{2K},
\label{eq:Rs-a1}
\end{align}
which implies $d_{\FB}(\INR=\SNR)\leq \frac{1}{K}$. 

However,  a more accurate relationship between $\INR$ and $\SNR$ has to be taken into account when $\SNR$ and $\INR$ are close to each other. The reason is that
\[
\lim_{\SNR\rightarrow \infty}\frac{\log \INR}{\log \SNR}=1
\]
may include several regimes with different capacity behaviors. For example, if $\INR=\beta \SNR$ with constant $\beta\notin \left(\frac{1}{2},2\right)$, we still have $\alpha=1$. Nevertheless, the result of Theorem~\ref{thm:Cap-Gaussian} holds for this regime of parameters, and thus $d_{\FB}=\frac{1}{2}$ can be achieved. An even more complicated scenario may happen when\footnote{Note that this regime is not included in the statement of Theorem~\ref{thm:Cap-Gaussian}. In fact, the gap between the rate can be achieved by the proposed coding scheme and the upper bound is not bounded for this regime. The characterization of (approximate) capacity remains as an open question.} $\INR=\SNR(1+o(\SNR))$ with $\lim_{\SNR\rightarrow \infty} o(\SNR)=0$. 

In other words, one has to be more careful when dealing with two simultaneous limiting behaviors, namely  $\log(\INR)\rightarrow \log(\SNR)$ and $\SNR\rightarrow \infty$, because depending on different rates of growth of $\SNR$ and convergence of $\INR$ to $\SNR$, different numbers of degrees of freedom can be achieved. This discontinuous behavior is similar to the discontinuity of the $\DoF$ of the fully connected interference channel (without feedback) studied in \cite{motahari2009degrees, etkin2009degrees}. It is shown in \cite{etkin2009degrees} that the per-user $\DoF$ of the $K$-user FC-IC is strictly less than $\frac{1}{2}$ when $\INR=\beta \SNR$ and $\beta$ is a non-zero rational coefficient. However, $\DoF=1/2$ can be achieved for irrational $\beta$. 

A slightly different (and perhaps more realistic) model to study $\DoF$ ($\GDoF$ at $\alpha=1$) is to fix the channel gains, and allow the transmit power of all transmitters to increase simultaneously, i.e., $P\rightarrow\infty$ where $\E[x_k^2]\leq P$. Under this model, instead of having two independently\footnote{Any rate of growth satisfying $\log \INR/\log \SNR \rightarrow 1$ is feasible in model in \eqref{eq:channel} and \eqref{eq:def:alpha}. } growing variables ($\SNR$ and $\INR$),  we deal with a single variable $P$, and the relationship between the signal-to-noise ratio and interference-to-noise ratio is controlled by the channel coefficients. A nice and generic result of Cadambe and  Jafar \cite{CJ:feedback-not-useful} shows that the per-user $\DoF$ of $K$-user FC-IC with  feedback and randomly chosen channel coefficients (not necessarily symmetric)  under the latter model is $1/2$, almost surely. 

\section{Gaussian Upper Bound for Global Feedback Model}
\label{sec:GF}
In Sections~\ref{sec:G-Ach} and \ref{sec:G-UB} we demonstrated the effect of \emph{local} feedback on the symmetric capacity of the $K$-user fully connected interference channel. It is shown that providing each transmitter with the signal observed by its receiver in the past can be significantly beneficial. In particular, it can improve the $\GDoF$ of the network for certain regimes of interference. A natural question arises is whether availability of more information through the feedback can further improve the symmetric capacity of the network. In the rest of this section we study the effect of \emph{global feedback} on the capacity of the network, that is when each transmitter has access to the received signals of not only its respective receiver, but all other receivers. We will show that for the symmetric topology which is of interest in this paper, global feedback does \emph{not} improve symmetric capacity beyond the local feedback (in approximate sense). This generalizes  the result of  \cite{sahai2011effective}  that in $2$-user interference channel with local feedback providing additional feedback link does not improve the capacity. 

In this model, the transmit signal of each user may depend on its message and all received signals in the past. Hence, in general we have
\begin{align}
x_{kt}=g_{kt} (W_k, y_1^{t-1},\dots,y_{K}^{t-1}).
\label{eq:gfb}
\end{align}
\begin{thm}
The symmetric  capacity of the $K$ user interference channel with global feedback with $\frac{\INR}{\SNR} \notin (0.5,2)$ can be approximated by 
\begin{align}
\Cs = \frac{1}{4} \log (1+ \SNR + \INR )+ \frac{1}{4} \log \left(1+\frac{\SNR}{1+\INR}\right). 
\end{align}
\end{thm}

Note that the coding scheme presented in Section~\ref{sec:G-Ach} well suits this model, and can be applied to achieve similar rates. We only need to derive an upper bound on the symmetric capacity for the global feedback model. 

Recall the proof of the upper bound of Theorem~\ref{thm:Cap-Gaussian} in Section~\ref{sec:G-UB}.  It is clear that the initial bound  in  \eqref{eq:UB-GS-1} is still valid, regardless of the feedback model. Moreover, bounding inequalities \eqref{eq:term1} and \eqref{eq:term3},  used to bound the first and third terms in \eqref{eq:UB-GS-1} respectively, would remain the same under the global feedback model. However, the argument we used to bound the second bound in \eqref{eq:term2} is not valid any more. The reason is that step $(c)$ in \eqref{eq:y_i-reconstruct} does not hold for global feedback model, because in this model the input signal $x_{jt}$ depends not only on $(W_j, y_j^{t-1})$, but also on $(y_1^{t-1}, y_2^{t-1},\dots, y_K^{t-1})$, and $y_1^{t-1}$ is missing in the condition. Alternatively, we can use the following lemma. 

\begin{lm}
For any reliable coding scheme of block length $T$, the transmit signal of  users $k=2,3,\dots,K$ at time $t$ can be determined from \[Q_t=\{y_1^{t-1},y_2^T, \tz_3^T,\dots, \tz_K^T, W_2,W_3,\dots,W_K\},\] 
where $\tz_{kt}=z_{kt}-z_{2t}$ for $k=3,4,\dots,K$. More precisely, for any family of coding functions $\{g_{kt}\}$ defined in \eqref{eq:gfb}, there exist  corresponding coding functions $\{\g_{kt}\}$ such that
\[
x_{kt}=\g_{kt}(y_1^{t-1}, y_2^T, \tz_3^T,\dots, \tz_K^T, W_2,W_3,\dots,W_K)
\]
for $k=2,3,\dots,K$ and $t=1,\dots,T$.
\label{lm:encoding-global-FB}
\end{lm}
\begin{proof}
 We prove this claim by induction on $t$. For $t=1$, the claim is obvious, since there is no feedback in the system and $x_{k1}=g_{k1}(W_k)$.  Assume the claim is correct for $t=\ell-1$. We will show that the claim valid for $t=\ell$, i.e., $x_{k\ell}=\g_{k\ell}(Q_{\ell})=\g_{k\ell}(y_{1,\ell-1}, Q_{\ell-1} ) $. 
Similar to \eqref{eq:connect-yk-to-y2}, we have
\begin{align*}
y_{k,\ell-1}\hspace{-1pt}=\hspace{-1pt}y_{2,\ell-1} \hspace{-1pt}+\hspace{-1pt}(\sqrt{\SNR} \hspace{-1pt}-\hspace{-1pt}\sqrt{\INR}) (x_{k,\ell-1}\hspace{-1pt}-\hspace{-1pt}x_{2,\ell-1}) \hspace{-1pt}+\hspace{-1pt}\tz_{2,\ell-1}.
\end{align*}
Note that $y_{2,\ell-1} $ and $\tz_{2,\ell-1}$ are given in $Q_{\ell-1}$. Moreover, $x_{k,\ell-1}=\g_{k,\ell-1}(Q_{\ell-1})$ and 
$x_{2,\ell-1}=\g_{2,\ell-1}(Q_{\ell-1})$. Hence, we can find $y_{k,\ell-1}$ from $Q_{\ell-1}$ for $k=2,\dots,K$. This means the output of all receivers at time $\ell-1$ (except $\r{1}$) is known from $Q_{\ell-1}$ (and hence from $Q_{\ell}$). On the other hand the channel output for $\r{1}$ at time $\ell-1$ is explicitly given in $Q_{\ell}$. Therefore, all the  channel outputs $y_{kt}$ for $k=1,\dots,K$ and $t=1,\dots,\ell-1$ are given by $Q_{\ell}$, which together with $W_k$ can uniquely determine the transmit signals $x_{2\ell},\dots,x_{K\ell}$.
\end{proof}

Now, from Lemma~\ref{lm:encoding-global-FB}, we can bound the second term in \eqref{eq:UB-GS-1} as follows.
\begin{align}
h\Big(\hspace{-1pt}y_1^T \Big| y_2^T, &\tz_{[3:K]}^T , W_{[2:K]} \hspace{-1pt}\Big)  \hspace{-1pt}=\hspace{-1pt} \sum_{t=1}^T \hspace{-1pt} h\Big(\hspace{-1pt}y_{1t} \Big| y_1^{t-1}, y_2^T, \tz_{[3:K]}^T , W_{[2:K]} \hspace{-1pt}\Big) \nonumber\\
&\stackrel{(a)}= \sum_{t=1}^T h\Big(y_{1t} \Big| y_1^{t-1}, y_2^T, \tz_{[3:K]}^T , W_{[2:K]} , x_{[2:K]}^T \Big)\nonumber \\
&\leq\frac{T}{2} \log \left(1+\frac{\SNR}{1+\INR}\right)+\frac{T}{2} \log (2\pi e)
\end{align}
where we used  Lemma~\ref{lm:encoding-global-FB} in $(a)$, and the last inequality follows the same argument used in \eqref{eq:term2}.

The rest of the proof is similar that in Section~\ref{sec:G-UB}, since all the other inequalities still hold under the global feedback. We skip the details in order to  avoid repetition. 

\section{Conclusion}
\label{sec:con}
We have studied the feedback capacity of the fully connected $K$-user interference channel under a symmetric topology.  This is a natural extension of the feedback capacity characterization for the $2$-user case in \cite{SuhTse11}, in which it is shown that channel output feedback can significantly improve the performance of the $2$-user interference channel. Rather surprisingly, it turns out that such an improvement can also be achieved  in the $K$-user case, except if the intended and interfering signals have the same received power at the receivers. In particular, we have shown that the per-user feedback capacity of the $K$-user FC-IC is as if there were only one source of interference in the network. Compared to the network without feedback \cite{JV:K-IF:10}, this result shows that feedback can significantly improve the network capacity. 

The coding scheme used to achieve the capacity of the network combines two well-known interference management techniques, namely, interference alignment and interference decoding. In fact, the messages at the transmitters are encoded such that the $K-1$ interfering signals are received aligned at each receiver. Closedness of lattice codes with respect to summation implies that the aligned received interference is a codeword that can be decoded, as in the $2$-user case. Another interesting aspect of this scheme is that each  message is kept secret from all receivers, except the intended one. This implies that an appropriately defined secrecy capacity of the network coincides with the capacity with no secrecy constraint. 

\appendices

\section{Coding Schemes for the Deterministic Network: Arbitrary $(n,m)$}
\label{app:det-ach}
\subsection{Weak Interference Regime ($m<n$)}
\label{app:det-ach-weak}
In the following, we generalize the coding scheme presented in Fig.~\ref{fig:det_weak} for arbitrary parameters $m$ and $n$. In the following, we use $A(a:b)$ to denote a column vector $\begin{bmatrix}A(a) & A(a+1) & \cdots & A(b)\end{bmatrix}^\prime$, where $a\leq b$ are two positive integer numbers. 
 
Denote the message of user $k$ which will be transmitted in $2$ channel uses by a $p$-ary sequence of length $2 \Rs$, namely, $\bS_k\triangleq S_k(1:2n-m)=\left[S_{k}(1), \dots, S_k(2n-m) \right]^\prime$, where $[\cdot]^\prime$ denote matrix transpose. Each user sends $p=n$ fresh symbols over its first channel use, i.e., 
\[
X_{k1}=S_k(1:n) =\begin{bmatrix}
S_k(1) & S_k(2) & \cdots & S_k(n)\end{bmatrix}^\prime.
\]
The signal received at the $\r{k}$ can be split into two parts, the part above the interference level which contains $(n-m)$  interference free symbols, and the lower $m$ symbols which is a combination of the intended symbols and interference, 
\begin{align*}
Y_{k1}\hspace{-1pt}&=\hspace{-1pt}\Big[
S^\prime_{k}(1\hspace{-1pt}:\hspace{-1pt}n-m), \nonumber\\
&\hspace{32pt} S_k(n-m+1)+S_{\sim k}(1), \dots, S_k(n)+ S_{\sim k}(m)\Big]^\prime\hspace{-2pt},
\end{align*}
where $S_{\sim k}(j)=\sum_{i\neq k} S_i(j)$ is the summation of all $p$-ary symbols sent by all the base stations except  $\t{k}$.
This received signal is sent to the transmitter via the feedback link. Transmitter $\t{k}$ first removes its own signal from this feedback signal, and then forwards the remaining symbols on its top most $m$ levels. It also transmits $(n-m)$ new fresh symbols over its lower levels:
\[
X_{k2}=\Big[
 S_{\sim k}(1), \dots, S_{\sim k}(m) , \ S^\prime_k(n+1:2n-m)\Big]^\prime.
\]
A similar operation is performed at all other transmitters, which results in a received signal at $\r{k}$ of the form 
\begin{align*}
Y_{k2} &= X_{k2}+ D^{n-m}\sum_{i\neq k} X_{i2}\nonumber\\
&= \begin{bmatrix}
S_{\sim k}(1) \\
\vdots \\
S_{\sim k}(m) \\
S_{k}(n+1) \\
\vdots \\
S_{k}(2n-m) \\
\end{bmatrix}
+ 
\begin{bmatrix}
0 \\
\vdots \\
0 \\
\sum_{i\neq k }S_{\sim i}(1) \\
\vdots \\
\sum_{i\neq k}S_{\sim i}(m) \\
\end{bmatrix}\nonumber\\
&=\hspace{-3pt}
\begin{bmatrix}
S_{\sim k}(1) \\
\vdots \\
S_{\sim k}(m) \\
S_{k}(n+1) \\
\vdots \\
\hspace{-1pt}S_{k}(2n-m) \hspace{-1pt}\\
\end{bmatrix}
\hspace{-3pt}+\hspace{-1pt} 
(K\hspace{-2pt}-\hspace{-2pt}1)\hspace{-1pt}\begin{bmatrix}
0 \\
\vdots \\
0 \\
 S_k(1) \\
\vdots \\
\hspace{-1pt}S_k(m)\hspace{-1pt}\\
\end{bmatrix}
\hspace{-2pt}+\hspace{-2pt} (K\hspace{-2pt}-\hspace{-2pt}2)\hspace{-2pt}
\begin{bmatrix}
0 \\
\vdots \\
0 \\
S_{\sim k}(1) \\
\vdots \\
\hspace{-1pt}S_{\sim k}(m)\hspace{-1pt}\\
\end{bmatrix}\hspace{-1pt}.
\end{align*}
We used the fact that 
\[
\sum_{i\neq k}\nolimits S_{\sim i}(j)= (K-1)S_k(j) + (K-2)S_{\sim k}(j) 
\]
in the last equality.
Having $Y_{k1}$ and $Y_{k2}$, receiver $\r{k}$ wishes to decode  $\bS_k$. Note that we have a linear system with $2n$ equations and $2n$ variables (including $m$ variables $S_{\sim k}(j)$ for $j=1,\dots, m$ and $2n-m$ variables including $S_k(j)$ for $j=1,\dots,2n-m$), which can be uniquely solved\footnote{It is easy to verify that the coefficient matrix is full-rank.}. Therefore, $\r{k}$ can recover all its $2n-m$ symbols transmitted by $\t{k}$, which implies a communication rate of $R_k=(2n-m)/2$. Note that the encoding operations at all transmitters are the same, and hence, a similar rate can be achieved for all pairs by applying a similar decoding.

\subsection{Strong Interference Regime ($m>n$)}

Similar to the weak interference regime, this scheme is performed over two consecutive time instances, and provides a total of $m$ information symbols for each user. Denote the message of user $k$ by a $\bS_k=\left[S_k(1),\dots, S_k(m)\right]$, which is a $p$-ary sequence of length $m$. In the first time instance, each user broadcasts its entire message, 
\[
X_{k1}=S_k(1:m)=\begin{bmatrix}
S_k(1) & \cdots & S_k(m)
\end{bmatrix}^\prime,
\]
which implies the received signal at $\r{k}$ to be
\begin{align*}
Y_{k1}&=D^{m-n} \bS_k + \bS_{\sim k}\nonumber\\
&=\hspace{-2pt}
\Big[
S_{\sim k}(1), \dots , S_{\sim k}(m-n)  , \nonumber\\
&\hspace{30pt} S_k(1)+S_{\sim k}(m-n+1) ,\dots , S_k(n) + S_{\sim k}(m) 
\Big]^\prime\hspace{-2pt}.
\end{align*}
This output is sent to the transmitter through the feedback link. In the second time slot, the transmitter simply removes its signal and forwards the remaining, that is,
\begin{align*}
X_{k2}=\begin{bmatrix}
S_{\sim k}(1) & \dots & S_{\sim k}(m) 
\end{bmatrix}^\prime.
\end{align*}
Hence, we have
\begin{align*}
Y_{k2}&=
\begin{bmatrix}
0\\
\vdots\\
0\\
S_{\sim k}(1)\\
\vdots\\
S_{\sim k}(n)
\end{bmatrix}
+
\begin{bmatrix}
\sum_{i\neq k} S_{\sim i}(1) \\
 \\
\vdots \\
\\
\\
\sum_{i\neq k} S_{\sim i}(m)
\end{bmatrix}\nonumber\\
&=\hspace{-1pt}
\begin{bmatrix}
0\\
\vdots\\
0\\
S_{\sim k}(1)\\
\vdots\\
S_{\sim k}(n)
\end{bmatrix}
\hspace{-1pt}+\hspace{-1pt}
(K\hspace{-1pt}-\hspace{-1pt}1)\hspace{-1pt} 
\begin{bmatrix}
S_k(1)\\
\\
\vdots\\
\\
\\
S_{k}(m)
\end{bmatrix}
\hspace{-1pt}+\hspace{-1pt}(K\hspace{-1pt}-\hspace{-1pt}2)\hspace{-1pt}
\begin{bmatrix}
S_{\sim k}(1)\\
\\
\vdots\\
\\
\\
S_{\sim k}(m)\\
\end{bmatrix}\\
&=(K-1) \bS_k+(D^{m-n}+(K-2) \mathbf{I}) \bS_{\sim k},
\end{align*}
where $\mathbf{I}$ is the identity matrix of proper size ($m\times m $ in the equation above). Having $Y_{k1}$ and $Y_{k2}$ together, $\r{k}$ has a linear system with $2m$ equation and $2m$ variables (including $m$ variables in $\bS_k$ and $m$ variables in $\bS_{\sim k}$):
\begin{align*}
\begin{bmatrix}
Y_{k1} \\ Y_{K2}
\end{bmatrix}
= 
\begin{bmatrix}\begin{array}{c|c}
D^{m-n} & \mathbf{I}\\
\hline
(K-1)\mathbf{I} & D^{m-n}+(K-2) \mathbf{I}
\end{array}
\end{bmatrix}
\begin{bmatrix}
\bS_k\\
\bS_{\sim k}
\end{bmatrix}.
\end{align*}
This system has a unique solution if and only if the coefficient matrix is full-rank. Note that  $K\not\equiv  1 (\mod q)$ is a necessary and sufficient condition for having a unique solution, which can be easily satisfied for a proper choice\footnote{Note that this result does not necessarily holds for all values of $p$ and $K$. For instance, this approach does not give a set of independent linear equations for the $3$-user case over the binary field. However, the encoding scheme for larger field size ($p>2$) still reveals valuable insights for  the Gaussian channel.} of $p$. Fig.~\ref{fig:det-strong} pictorially demonstrates this coding scheme for $3$-user case.

\section{Quasi-symmetric Fully Connected $K$-user Interference Channel under the Deterministic Model}
\label{app:det-qsym}
Note that the symmetry of the channels in the fully-symmetric model allows us to align the interfering signals in the second block and easily reconstruct the same interfering signal as in the first block at each receiver. This is not possible in the quasi-symmetric case, since it is not clear whether one can simultaneously align all interfering signals.

In the following we present a generic coding scheme together with a sufficient condition which guarantees feasibility of simultaneous alignment for the quasi-symmetric model. We will further show that this conditions holds for any choice of channel signs for $K=3$.
This scheme is based on opportunistically choosing the coefficient of the feedback signal in formation of the transmit signal in the second phase. 

\begin{lm}
Simultaneous alignment of interference in the quasi-symmetric fully connected $K$-user interference channel is feasible provided that there exist (non-zero) diagonal matrices $A$, $B$, $U$ and $V$ such that 
\begin{align}
\Lambda A + \Lambda B \Lambda = U + V \Lambda, 
\label{eq:suf-cond}
\end{align}
where $\Lambda=\{\lambda_{ij}\}$ is the network sign matrix with zero diagonal elements, and $\{\pm 1\}$ off-diagonal elements. 
\label{lm:suf-cond}
\end{lm}
In the following we prove this lemma, and at the end of this section we show that \eqref{eq:suf-cond} can be satisfied with diagonal matrices for any $\Lambda$ of size $3\times 3$.

\begin{proof}[Proof of Lemma~\ref{lm:suf-cond}]
Assume diagonal matrices $A$, $B$, $U$ and $V$ exist such that \eqref{eq:suf-cond} holds. We present a coding scheme which guarantees simultaneous interference alignment at all the receivers. \\

\noindent \underline{Case I: Weak Interference Regime ($n>m$)}

We borrow the notation from Appendix~\ref{app:det-ach}, to denote the transmit signal of each user in the first block of transmission:
\[
X_{k1}=S_k(1:n) =\Big[
S_k(1) , S_k(2) , \dots , S_k(n)\Big]^\prime.
\]
Upon receiving feedback, transmitter $\t{k}$ can subtract its own contribution from the received signal  at $\r{k}$ and compute
\begin{align*}
\mathcal{I}_k=
\begin{bmatrix}
I_k(1) \\ I_k(2) \\ \vdots \\ I_k(m)
\end{bmatrix}
\triangleq 
\begin{bmatrix}
\sum_i \lambda_{ki} S_i(1) \\ 
\sum_i \lambda_{ki} S_i(2)\\
\vdots\\
\sum_i \lambda_{ki} S_i(m)
\end{bmatrix}.
\end{align*}
The transmit signal of user $k$ in the second block, will be formed as 
\begin{align*}
X_{k2}=\begin{bmatrix}
A_k S_{k}(1) + B_k I_k(1)\\
\vdots\\
A_k S_{k}(m) + B_k I_k(m)\\
A_k S_k(m+1) + B_k S_k(n+1)\\
\vdots\\
A_k S_k(n) + B_k S_k(2n-m)
\end{bmatrix},
\end{align*}
where $A_k$ and $B_k$ are the $k$th diagonal elements of matrices $A$ and $B$, respectively. Performing such coding scheme at each transmitter, the received signal at $\r{k}$ in the second block can be written as
\begin{align}
Y_{k2}&= X_{k2}+ D^{n-m}\sum_{i} \lambda_{ki} X_{i2}\nonumber\\
&= X_{k2}+ \sum_{i} \lambda_{ki} \begin{bmatrix}
0\\
\vdots\\
0\\
 A_i S_i(1)+ B_i \sum_j \lambda_{ij} S_j(1) \\
 \vdots\\
 A_i S_i(m)+ B_i \sum_j \lambda_{ij} S_j(m) 
\end{bmatrix}
\nonumber\\
&= 
X_{k2} + \begin{bmatrix}
0\\
\vdots\\
0\\
\sum_i [\lambda_{ki} A_i + \sum_{j} \lambda_{kj} B_j \lambda_{ji}] S_i(1)  \\
\vdots\\
\sum_i [\lambda_{ki} A_i + \sum_{j} \lambda_{kj} B_j \lambda_{ji}] S_i(m)
\end{bmatrix}
\label{eq:y2-q-sym}
\end{align} 
Now note that $\lambda_{ki} A_i$ and $\sum_{j} \lambda_{kj} B_j \lambda_{ji}$ are $(k,i)$th elements of matrices $\Lambda A$ and $\Lambda B \Lambda$, respectively, and from \eqref{eq:suf-cond} we have
\begin{align*}
\lambda_{ki} A_i + \sum_{j} \lambda_{kj} B_j \lambda_{ji}&= (\Lambda A + \Lambda B \Lambda)_{ki} \\
&= (U+V \Lambda)_{ki}= U_{ki} + V_k \lambda_{ki}.
\end{align*} 
Therefore, 
\begin{align}
\sum_i \Big[\lambda_{ki} A_i &+ \sum_{j} \lambda_{kj} B_j \lambda_{ji} \Big]  S_i(\ell) =  
\sum_i [U_{ki} + V_k \lambda_{ki}] S_i(\ell)\nonumber\\
&= U_k S_k(\ell) + V_k \sum_{i} \lambda_{ki} S_i(\ell)\nonumber\\
&= U_k S_k(\ell) + V_k  I_k(\ell),
\label{eq:q-sym-1}
\end{align}
for $\ell=1,2,\dots,m$. Replacing this into \eqref{eq:y2-q-sym}, we get
\begin{align*}
Y_{k2}\hspace{-1pt}=\hspace{-3pt} \begin{bmatrix} 
A_{k} S_k(1) + B_k I_k(1)\\
 \vdots\\
 A_{k} S_k(m) + B_k  I_k(m)\\
A_k S_k(m\hspace{-1pt}+\hspace{-1pt}1) \hspace{-1pt}+\hspace{-1pt} B_k S_k(n\hspace{-1pt}+\hspace{-1pt}1)\hspace{-2pt}\\
 \vdots\\
A_k S_k(n) + B_k S_k(2n-m)  
\end{bmatrix}
\hspace{-3pt}+\hspace{-3pt} \begin{bmatrix}
0\\
\vdots\\
0\\
U_k S_k(1) + V_k I_k(1)  \\
\vdots\\
\hspace{-1pt}U_k S_k(m) \hspace{-1pt}+\hspace{-1pt} V_k I_k(m)\hspace{-2pt}
\end{bmatrix}
\end{align*}

Therefore, all the interfering signals received at $\r{k}$ in the second block are scaled version of those received in the first block, and so it has to decode $2n-m$ information symbols $S_k(1),\dots,S_k(2n-m)$  as well as $m$ interference symbols $I(1),\dots,I(m)$ based on its $2n$ equations. This system is given by
\begin{align*}
\begin{bmatrix}
Y_{k1}\\Y_{k2}
\end{bmatrix}
&=\left[\begin{array}{c|c}
\mathbf{I} & D^{n-m}\\
\hline
A_k \mathbf{I} + U_k D^{n-m}  & B_k \mathbf{I} + V_k D^{n-m}
\end{array}\right]\nonumber\\
&\hspace{100pt}\times
\left[\begin{array}{c}
S_{k}(1:n) \\
\hline
I_k(1:m)\\
\hline
S_k(n+1:  2n -m)
\end{array}
\right]\hspace{-2pt}.
\end{align*}
One can easily verify that the determinant of the coefficient matrix of this system of equations is
\begin{align*}
\det(\mathbf{I}) &\det\left( B_k \mathbf{I} + V_k D^{n-m} - (A_k \mathbf{I} + U_k D^{n-m}) \mathbf{I}^{-1} D^{n-m} \right)\nonumber\\
&= \det ( B_k \mathbf{I} + (V_k-A_k) D^{n-m} - U_k D^{2(n-m)})\\
&=B_k^K,
\end{align*}
and hence it is full-rank and so decoding is feasible provided that $B_k\neq 0$ for $k=1,\dots,K$.  \\

\noindent \underline{Case II: Strong Interference Regime ($n<m$)}\\
Similar to the last case, each transmitter first sends $m$ information symbol in the first block. Upon receiving the feedback,  transmitter $\t{k}$ subtract its own contribution to find the interfering symbols at $\r{k}$
\begin{align*}
\mathcal{I}_k=
\begin{bmatrix}
I_k(1) \\ I_k(2) \\ \vdots \\ I_k(m)
\end{bmatrix}
\triangleq 
\begin{bmatrix}
\sum_i \lambda_{ki} S_i(1) \\ 
\sum_i \lambda_{ki} S_i(2)\\
\vdots\\
\sum_i \lambda_{ki} S_i(m)
\end{bmatrix}.
\end{align*}
It forms a linear combination of its own symbols and the interfering symbols, and transmits the result over the second block:
\begin{align*}
X_{k2}=\begin{bmatrix}
A_k S_{k}(1) + B_k I_k(1)\\
\vdots\\
A_k S_{k}(m) + B_k I_k(m)\\
\end{bmatrix}.
\end{align*}
Using \eqref{eq:q-sym-1}, the received signal at $\r{k}$ in the second block can be written as \eqref{eq:long-1} given at the top of next page.
\begin{figure*}
\begin{align}
Y_{k2}&= D^{m-n} X_{k2}+ \sum_{i} \lambda_{ki} X_{i2}
= D^{m-n} X_{k2}+ \sum_{i} \lambda_{ki} \begin{bmatrix}
 A_i S_i(1)+ B_i \sum_j \lambda_{ij} S_j(1) \\
 \vdots\\
 A_i S_i(m)+ B_i \sum_j \lambda_{ij} S_j(m) 
\end{bmatrix}\nonumber\\
&=  
\begin{bmatrix}
0 \\
\vdots\\
0\\
A_k S_k(1) + B_k I_k(1)\\
\vdots\\
A_k S_k(n) + B_k I_k(n)
\end{bmatrix}
+ \begin{bmatrix}
U_k S_k(1) + V_k I_k(1) \\
\vdots\\
U_k S_k(m-n) + V_k I_k(m-n) \\
U_k S_k(m-n+1) + V_k I_k(m-n+1) \\
\vdots\\
U_k S_k(m) + V_k I_k(m)
\end{bmatrix}
\label{eq:long-1}
\end{align} 
\hrule
\end{figure*}
Note that this is a function of only $S_k(1),\dots,S_k(m)$ and $I_k(1),\dots,I_k(m)$ variables; so the receiver can can solve the system of available $2m$ equations for these variables. More precisely, it has to solve
\begin{align*}
\begin{bmatrix}
Y_{k1}\\
Y_{k2}
\end{bmatrix}\hspace{-2pt}
= \hspace{-2pt}
\left[ 
\begin{array}{c|c}
 D^{m-n} &\mathbf{I} \\
\hline
\hspace{-3pt} U_k \mathbf{I} + A_k D^{m-n} \hspace{-2pt}& \hspace{-2pt}V_k \mathbf{I} + B_k D^{m-n}\hspace{-3pt}
\end{array}
\right]\hspace{-2pt}
\left[\begin{array}{c}
\hspace{-2pt}S_k(1:m)\hspace{-2pt}\\
\hline
\hspace{-2pt}I_{k}(1:m)\hspace{-2pt} 
\end{array}\right].
\end{align*}
It is easy to verify that the coefficient matrix is non-singular, since its determinant can be written as
\begin{align*}
\det(\mathbf{I}) &\det\left( (U_k \mathbf{I} + A_k D^{m-n}) \hspace{-1pt}- \hspace{-1pt}(V_k \mathbf{I} + B_k D^{m-n}) \mathbf{I}^{-1} D^{m-n}  \right)\nonumber\\
 &= \det\left( U_k\mathbf{I} +(A_k-V_k) D^{m-n} \hspace{-1pt}-\hspace{-1pt}B_k D^{2(m-n)}\right)\\
 &=U_k^K
\end{align*}
which is non-zero, provided that $U_k\neq 0$ for $k=1,2,\dots,K$. \\

\noindent \underline{Case III: Moderate Interference Regime ($n=m$)}\\
The coding scheme and argument for this case is the same as the previously discussed cases. At the end of the second channel use, $\r{k}$ has to solve 

\begin{align*}
\begin{bmatrix}
Y_{k1}\\
Y_{k2}
\end{bmatrix}
= 
\left[ 
\begin{array}{c|c}
\mathbf{I} &\mathbf{I} \\
\hline
 (U_k+A_k) \mathbf{I}  & (V_k+B_k) \mathbf{I} 
\end{array}
\right]
\left[\begin{array}{c}
S_k(1:n)\\
\hline
I_{k}(1:n) 
\end{array}\right].
\end{align*}
for $S_k(1:n)$. It is clear that this system can be uniquely solved provided that $B_k-A_k+V_k+U_k\neq 0$ for $k=1,\dots,K$.
\end{proof}
Now we are ready to prove Theorem~\ref{thm:det-qsym} based on Lemma~\ref{lm:suf-cond}. 

\subsection{Proof of Theorem~\ref{thm:det-qsym}}
The proof of the converse part  is similar to that of Theorem~\ref{thm:det} for $m\neq n$, and hence we skip it. For $m=n$ we may distinguish the following two cases:

\paragraph*{Case I: $\Lambda+\mathbf{I}$ is full-rank}
In this case we can use the upper bound \eqref{eq:Rsym-m-neq-n} presented in Section~\ref{sec:det-OB}, which yields 
\[
\Rs \leq \max \left(\frac{n}{2}, n-\frac{n}{2}\right)=\frac{n}{2}.
\]
\paragraph*{Case II: $\Lambda+\mathbf{I}$ is singular}
It is easy to check that if $\rank(\Lambda+\mathbf{I})>0$. If $\rank(\Lambda+\mathbf{I})=1$, then $\lambda_{ij}=1$ for $i\neq j$, and therefore the argument in \eqref{eq:UB-n-equal-m} holds, which shows $\Rs\leq n/3$. 

Now, assume $\rank(\Lambda+\mathbf{I})=2$. It is easy to verify that a $3\times 3 $ matrix with elements in $\{\pm 1\}$ has $\rank=2$ if and only if it has two (up to negative sign) identical  rows.  Without loss of generality, we may assume the first and second rows are identical, which yields $Y_2{t}$ can be deterministically recovered from $Y_{1t}$. 

So, we can write
\begin{align*}
H(W_1&, W_2,W_3 | Y_1^T) = H(W_1,W_2,W_3 | Y_1^T, Y_2^T)\nonumber\\
&=   H(W_3 | W_1,W_2, Y_1^T, Y_2^T) +H(W_1,W_2 | Y_1^T, Y_2^T)\\
& \leq H(W_3 | W_1,W_2, Y_1^T, Y_2^T, X_1^T, X_2^T) +T\epsilon_T\\
&=H(W_3 | W_1,W_2, Y_1^T, Y_2^T, X_1^T, X_2^T, X_3^T, Y_3^T) +T\epsilon_T\nonumber\\
&\leq  2T\epsilon_T,
\end{align*}
and therefore,
\begin{align*}
T(R_1&+R_2+R_3) = H(W_1,W_2,W_3) \nonumber\\
&\leq I(W_1,W_2,W_3;Y_1^T) + H( W_1,W_2,W_3|Y_1^T)\nonumber\\
&\leq H(Y_1^T) + 2T\epsilon_T \leq nT +2T\epsilon_T
\end{align*}
which yields in $\Rs\leq n/3$.

In order to prove the achievability part of Theorem~\ref{thm:det-qsym} it remains to be shown that for any sign matrix of size $3\times 3$ there exist a solution for \eqref{eq:suf-cond}.
\begin{lm}
For any given sign matrix $\Lambda$ of size $3\times 3$, there exist diagonal matrices $A$, $B$, $U$, and $V$ such that 
\begin{align*}
\Lambda A + \Lambda B \Lambda = U + V \Lambda.
\end{align*}
and one of the following holds:
\begin{enumerate}
\item[$a)$] $B_k\neq 0, \ \forall k$ (for the weak interference regime); or
\item[$b)$]  $U_k\neq 0,\ \forall k$ (for the strong interference regime); or
\item[$c)$] $B_k-A_k+V_k-U_k\neq 0, \ \forall k$ (for the moderate interference regime).
\end{enumerate}
\label{lm:q-sym-K3}
\end{lm}
\begin{proof}[Proof of Lemma~\ref{lm:q-sym-K3}]
Note that for $K=3$ we are free to choose $12$ variables (the diagonal elements of matrices $A$, $B$, $U$ and $V$) such that they satisfy $9$ linear equations. It is clear that this system of equations has multiple solutions. Hence, depending in the interference regime of interest, we can choose $3$ variables in order to make the overall transition matrix of two channel uses full rank, i.e., we set $B_k\neq 0$ in the weak interference regime for $k=1,2,3$, and $U_k\leq 0$ for the strong interference regime. 

Regarding the moderate interference regime ($m=n$), the constraint to have a solvable system of equation is $B_k-A_k+V_k+U_k\neq 0$ for $k=1,2,3$. It can be shown that when $\Lambda+\mathbf{I}$ is a full-rank matrix, these additional constraints are feasible with the solution for matrices $A$, $B$, $U$, and $V$, and therefore $\Rs=n/2$ can be achieved using cooperative interference alignment. However, if $\Lambda+\mathbf{I}$ is a singular matrix, $\Rs=n/3$ can be simply achieved by time-sharing scheme. This completes the proof. 
\end{proof}

\section{Weak Secrecy Provided by Using Lattice Codes}
\label{app:secrecy}
In this section we prove the claim of Remark~\ref{rmk:secrecy}. We can break the output signal of $\r{j}$ into two blocks, and write 
\begin{align}
I(y_j^{2T} &; W_k)= I(y_{j1}^T, y_{j2}^T ; W_k) \leq I(y_{j1}^T, y_{j2}^T ; \bs_{k0}, \bc_{k1}, \bc_{k2})\nonumber\\
&\leq I(y_{j1}^T, y_{j2}^T, \bs_{j0}, \bc_{j1}, \bc_{j2} ; \bs_{k0}, \bc_{k1}, \bc_{k2})\nonumber\\
&= \hspace{-1pt} I\Bigg( \hspace{-2pt} \sum_{i\neq j} \bs_{i0} , \hspace{-1pt} \sum_{i\neq j} \bc_{i1}+\bz_{j1}, \hspace{-1pt} \sum_{i\neq j}  \bc_{i2} + \bz_{j2}  ;
\bs_{k0}, \bc_{k1}, \bc_{k2} \hspace{-2pt}\Bigg)\nonumber\\
&= I\Big(\sum_{i\neq j} \bs_{i0} ; \bs_{k0}\Big) + I\Big(\sum_{i\neq j} \bc_{i1}+\bz_{j1} ; \bc_{k1}\Big) \nonumber\\
&\phantom{=}+ I\Big(\sum_{i\neq j} \bc_{i2}+\bz_{j2} ; \bc_{k2}\Big)\label{eq:secrecy-1}
\end{align}
where \eqref{eq:secrecy-1} holds because the three pairs $(\sum_{i\neq j} \bs_{i0} , \bs_{k0})$, $(\sum_{i\neq j} \bc_{i1}+\bz_{j1} ; \bc_{k1})$ and $(\sum_{i\neq j} \bc_{i2}+\bz_{j2} ; \bc_{k2})$ are mutually independent. The first term in \eqref{eq:secrecy-1} is zero for $K\geq 3$ due to the crypto lemma \cite{forney2003role}. The second and third terms can be upper bounded using the mutual information expression for  Gaussian variables. Hence, 
\begin{align*}
I(y_{j}^{2T} ; W_k) \leq 2 \frac{T}{2} \log \left(1 +\frac{1}{K-1}\right)=T\log \frac{K}{K-1}
\end{align*}
which is constant with respect to $\SNR$. 

\section*{Acknowledgement}
We are grateful to the associate editor and the reviewers for their suggestions which led to
the enhancement of the scope and the readability of this work.

\bibliographystyle{IEEEtran}

\bibliography{Refs}

\begin{biographynophoto}
{Soheil Mohajer} received the B.Sc. degree in electrical
engineering from the Sharif University of
Technology, Tehran, Iran, in 2004, and the M.Sc.
and Ph.D. degrees in communication systems both
from Ecole Polytechnique F\'{e}d\'{e}rale de Lausanne
(EPFL), Lausanne, Switzerland, in 2005 and 2010,
respectively. He then joined Princeton University,
New Jersey, as a post-doctoral research associate. Dr.
Mohajer has been a post-doctoral researcher at the
University of California at Berkeley, since October
2011.

His research interests include network information theory, data compression, wireless communication,  and bioinformatics.
\end{biographynophoto}

\begin{biographynophoto}{Ravi Tandon} (S03, M09) received the B.Tech degree
in electrical engineering from the Indian Institute
of Technology (IIT), Kanpur in 2004 and
the Ph.D. degree in electrical and computer engineering
from the University of Maryland, College
Park in 2010. From 2010 until 2012, he was
a post-doctoral research associate with Princeton
University. In 2012, he joined Virginia Polytechnic
Institute and State University (Virginia Tech)
at Blacksburg, where he is currently a Research
Assistant Professor in the Department of Electrical
and Computer Engineering. His research interests are in network information
theory, communication theory for wireless networks and information theoretic
security.

Dr. Tandon is a recipient of the Best Paper Award at the Communication
Theory symposium at the 2011 IEEE Global Telecommunications Conference.
\end{biographynophoto}

\begin{biographynophoto}
{H. Vincent Poor} (S72, M77, SM82, F87) received
the Ph.D. degree in electrical engineering and computer
science from Princeton University in 1977.
From 1977 until 1990, he was on the faculty of the
University of Illinois at Urbana-Champaign. Since
1990 he has been on the faculty at Princeton, where
he is the Dean of Engineering and Applied Science,
and the Michael Henry Strater University Professor
of Electrical Engineering. Dr. Poor's research interests
are in the areas of stochastic analysis, statistical
signal processing and information theory, and their
applications in wireless networks and related fields including social networks
and smart grid. Among his publications in these areas are 
{\it Smart Grid Communications and Networking} 
(Cambridge University Press, 2012) and {\it Principles of 
Cognitive Radio} (Cambridge University Press, 2013).

Dr. Poor is a member of the National Academy of Engineering and the
National Academy of Sciences, a Fellow of the American Academy of
Arts and Sciences, and an International Fellow of the Royal Academy of
Engineering (U. K.). He is also a Fellow of the Institute of Mathematical
Statistics, the Optical Society of America, and other organizations. In 1990,
he served as President of the IEEE Information Theory Society, in 2004-07
as the Editor-in-Chief of these {\it Transactions}, and in 2009 as General Co-chair
of the IEEE International Symposium on Information Theory, held in Seoul,
South Korea. He received a Guggenheim Fellowship in 2002 and the IEEE
Education Medal in 2005. Recent recognition of his work includes the 2010
IET Ambrose Fleming Medal for Achievement in Communications, the 2011
IEEE Eric E. Sumner Award, the 2011 IEEE Information Theory Paper Award,
and honorary doctorates from Aalborg University, the Hong Kong University
of Science and Technology, and the University of Edinburgh.
\end{biographynophoto}

\end{document}